\documentclass[conference]{IEEE-conference-template-062824/IEEEtran}
\IEEEoverridecommandlockouts

\def\BibTeX{{\rm B\kern-.05em{\sc i\kern-.025em b}\kern-.08em
    T\kern-.1667em\lower.7ex\hbox{E}\kern-.125emX}}

\usepackage{breakurl}
\usepackage[hyphens]{url}

\usepackage{amsthm}
\usepackage{amsmath}
\usepackage{amssymb}
\usepackage{amsfonts}
\usepackage{caption,subcaption}
\usepackage{balance}
\usepackage[binary-units]{siunitx}
\usepackage{graphicx}
\usepackage{hyperref}
\usepackage{cleveref}
\usepackage{color, colortbl} 
\usepackage{multirow}
\usepackage{makecell}

\newtheorem{theorem}{Theorem}
\newtheorem{lemma}{Lemma}
\newtheorem{example}{Example}

\newif\ifpublish
\publishtrue

\ifpublish
    \newcommand{\fulltext}[1]{#1}
\else
    \newcommand{\fulltext}[1]{}
\fi

\newif\ifcomments
\commentsfalse

\ifcomments
\newcommand{\Modify}[1]{{\color{blue}{#1}}}
\else
\newcommand{\Modify}[1]{#1}
\fi

\newcommand{\Service}{\mathcal{S}}
\newcommand{\Replicas}{\mathcal{R}}
\newcommand{\Miners}{\mathcal{M}}
\newcommand{\Clients}{\mathcal{C}}

\newcommand{\n}[1]{\mathbf{n}_{#1}}
\newcommand{\f}[1]{\mathbf{f}_{#1}}

\newcommand{\Replica}[1][r]{\textsc{#1}}

\newcommand{\Miner}[1][m]{\textsc{#1}}

\newcommand{\MName}[1]{\textsc{#1}}

\newcommand{\BFT}{\textsc{BFT}}

\newcommand{\PoW}{\textsc{PoW}}
\newcommand{\PoS}{\textsc{PoS}}
\newcommand{\PoC}{\textsc{PoC}}

\newcommand{\PBFT}{\textsc{Pbft}}
\newcommand{\pbft}{\textsc{Pbft}}
\newcommand{\hotstuff}{{HotStuff}}

\newcommand{\ResDB}{\text{ResilientDB}}
\newcommand{\MAC}{\textsc{Mac}}
\newcommand{\DS}{\textsc{DS}}

\newcommand{\Algorand}{{Algorand}}

\newcommand{\Ethereum}{{Ethereum}}

\newcommand{\Diem}{{Diem}}
\newcommand{\Quorum}{{Quorum}}

\newcommand{\Diablo}{{Diablo}}

\newcommand{\block}{\mathfrak{B}}
\newcommand{\Bsize}{\sigma}

\newcommand{\Slice}[1]{\mathbb{S}_{#1}}
\newcommand{\TotalSlices}{{\bf u}}
\newcommand{\ShiftRound}{\bf{r}}
\newcommand{\Sequence}{k}

\newcommand{\BlockHeight}{b}

\newcommand{\Token}{\Psi}

\newcommand{\Certificate}{\mathfrak{C}}

\newcommand{\Transaction}[1][t]{\MakeUppercase{#1}}
\newcommand{\Message}[2]{\textsc{#1}(#2)}
\newcommand{\SignMessage}[2]{\langle#1\rangle_{#2}}

\newcommand{\Hash}[1]{\texttt{hash}(#1)}

\newcommand{\Nonce}{\eta}

\newcommand{\prev}{prev}

\newcommand{\abs}[1]{\lvert #1 \rvert}

\usepackage[noend]{algorithmic}

\usepackage{tikz,pgfplots,pgfplotstable}
\usetikzlibrary{arrows.meta,calc,decorations.pathreplacing,patterns,shapes}
\tikzset{
    dot/.append style={circle,scale=0.35,draw=black,fill=black},
    sdot/.append style={scale=0.45,draw=black,fill=black},
    label/.append style={align=center,font=\strut\footnotesize},
    >=Stealth,
    every edge/.append style={semithick},
    thread/.append style={align=center,draw,thick,rectangle,text width=1cm,text height=2ex,text depth=.25ex,minimum height=0.75cm,font=\strut},
}

\definecolor{colA}{RGB}{230,159,0}
\definecolor{colB}{RGB}{86,180,233}
\definecolor{colC}{RGB}{0,158,115}
\definecolor{colD}{RGB}{240,228,66}
\definecolor{colE}{RGB}{0,114,178}
\definecolor{colF}{RGB}{213,94,0}
\definecolor{colG}{RGB}{204,121,167}
\definecolor{colH}{RGB}{238,130,238}
\definecolor{colI}{RGB}{64,224,208}
\definecolor{colGrey}{RGB}{211,211,211}
\definecolor{colIvory}{RGB}{255,235,215}
\definecolor{colDeepPink}{RGB}{255,20,147}
\definecolor{colLightRed}{RGB}{255,235,0}
\definecolor{colSkyBlue}{RGB}{135, 206, 235}
\definecolor{colVermillion}{RGB}{227, 66, 52}
\definecolor{colReddishPurple}{RGB}{100, 0, 120}
\definecolor{DarkOrange}{RGB}{255,140,0}
\definecolor{Peru}{RGB}{128,0,0}
\definecolor{deepblue}{rgb}{0,0,0.5}
\definecolor{deepred}{rgb}{0.6,0,0}
\definecolor{lightgreen}{RGB}{34,139,34}
\definecolor{lightyellow}{RGB}{218,112,214}
\definecolor{slateblue}{RGB}{123,104,238}
\definecolor{lightblue}{RGB}{30,144,255}
\definecolor{lightBrown}{RGB}{188,143,143}
\definecolor{deepGreen}{RGB}{0,128,0}
\definecolor{deepBlue}{RGB}{0,0,255}
\definecolor{deepPurple}{RGB}{128,0,128}
\definecolor{fakeGreen}{RGB}{102,205,170}
\definecolor{Maroon}{RGB}{210,105,30}

\pgfplotscreateplotcyclelist{mycyclelist}{
    very thick,solid,black,every mark/.append style={solid},mark=*\\         
    very thick,solid,colG,every mark/.append style={solid},mark=*\\          
    very thick,solid,colB,every mark/.append style={solid},mark=*\\          
    very thick,solid,colE,every mark/.append style={solid},mark=*\\          
    very thick,solid,colD,every mark/.append style={solid},mark=*\\          
}

\pgfplotsset{
    compat=1.16,
    width=195pt,
    height=156pt,
    every axis title shift=0pt,
    max space between ticks=25,
    every axis/.append style={
            cycle list name=mycyclelist,
            ymin=0,
            enlargelimits=0.05,
            scale ticks above exponent=1,
            scaled x ticks=false,
            xtick=data,
            mark size=1pt,
            font=\Large,
            y tick label style={
                },
            ylabel shift={-5pt}
        },
    every axis legend/.append style={
            cells={anchor=west}
        }
}

\tikzset{
    plot/.append style={baseline,scale=0.65}
}

\usepackage{tikz,pgfplots,pgfplotstable}
\usetikzlibrary{arrows.meta}
\tikzset{
    >=Stealth,
    dot/.style={circle,scale=0.35,draw=black,fill=black},
    node_text/.append style={font=\strut\bfseries},
    plot/.append style={baseline,scale=0.8},
    label/.append style={font=\strut\footnotesize}
}
\pgfplotscreateplotcyclelist{mycyclelist}{
    thick,black   ,every mark/.append style={solid,fill=\pgfplotsmarklistfill},mark=*\\
    thick,red     ,every mark/.append style={solid,fill=\pgfplotsmarklistfill},mark=square*\\
    thick,blue    ,every mark/.append style={solid,fill=\pgfplotsmarklistfill},mark=triangle*\\
    thick,teal    ,every mark/.append style={solid,fill=\pgfplotsmarklistfill},mark=pentagon*\\
    thick,brown   ,every mark/.append style={solid,fill=\pgfplotsmarklistfill},mark=halfsquare*\\
    thick,orange  ,every mark/.append style={solid,fill=\pgfplotsmarklistfill},mark=halfcircle*\\
    thick,violet  ,every mark/.append style={solid,fill=\pgfplotsmarklistfill,rotate=180},mark=halfdiamond*\\
}
\pgfplotscreateplotcyclelist{mycyclelistex}{
    black   ,every mark/.append style={solid,fill=\pgfplotsmarklistfill},mark=*\\
    red     ,every mark/.append style={solid,fill=\pgfplotsmarklistfill},mark=square*\\
}

\pgfplotsset{
scaled y ticks = false,
every axis/.append style={
        ylabel near ticks,
        xlabel near ticks,
        mark size=2.5pt,
        cycle list name=mycyclelist,
        y tick label style={
                /pgf/number format/precision=1,
                /pgf/number format/fixed,
                /pgf/number format/fixed zerofill
            }
    },
barstyle/.append style={
        ybar,
       bar width={0.5cm},
        enlarge x limits=0.4,
        enlarge y limits={upper=0.025},
        ymin=0,
        xtick=data
    }
every axis/.append style={
ylabel near ticks,
xlabel near ticks,
mark size=1pt,
cycle list name=mycyclelist,
enlargelimits=0.1
},
}

\newcommand{\axistput}{Throughput (\si{txn\per\second})}
\newcommand{\axistputblock}{Throughput (\si{txn\per\second})}
\newcommand{\axisnodes}{Number of Miners}

\newcommand{\axislat}{Latency (\si{\second})}
\newcommand{\axispbftlat}{Latency (\si{\second})}
\newcommand{\axisreplicas}{Number of Replicas}

\newcommand{\pbfttput}{\begin{tikzpicture}[plot]
        \begin{axis}[xlabel={\axisreplicas},ylabel={\axistput}, title={(c)},
                xticklabels={16,32,64,128},
                legend style={font=\footnotesize},
                width=200,
                height=120,
                legend style={font=\footnotesize, at={(0.36,0.68)}, anchor=north,legend columns=1},
                yticklabels={3k,10k,100k, 430k,900k},
                ytick={11.5,14,16.7, 18.4,19.7},
                ymin=11.5,
            ]
            \addplot[color=blue, mark=square] table[x={Id},y={logtps},  color=blue, mark=square] {\dataPoCLatencyPBFTX};
        \end{axis}
    \end{tikzpicture}}

\newcommand{\pbftlat}{\begin{tikzpicture}[plot]
        \begin{axis}[xlabel={\axisreplicas},ylabel={\axispbftlat}, title={(d)},
                xticklabels={16,32,64,128},
                legend style={font=\footnotesize},
                width=200,
                height=120,
                legend pos={north west}
            ]
            \addplot[color=blue, mark=square] table[x={Id},y={Latency},  color=blue, mark=square] {\dataPoCLatencyPBFTX};
        \end{axis}
    \end{tikzpicture}}

\pgfplotstableread{
    Id Nodes Latency Throughput CPU logtps
    1 16 0.072560329 892890  50     19.768122927
    2 32 0.248858182 453360  48     18.790297582
    3 64 0.514710447 222475  50     17.763283701
    4 120  0.997482 107126 60       16.708949146
}\dataPoCLatencyPBFT

\pgfplotstableread{
    Id Nodes Latency Throughput CPU logtps 
    1 16 0.172560329 922890  50    19.815799176574707   
    2 32 0.348858182 503360  48    18.94123105030305   
    3 64 0.844710447 243475  50    17.89341411849943   
    4 120 1.27315 121126 60         16.886149050822972 
}\dataPoCLatencyPBFTX

\pgfplotstableread{
    Id Nodes Latency Throughput CPU logtps
    1 16 1.688 692666  0            19.401800334
    2 32 1.506 551990  0            19.074282605
    3 64 2.192 298484  0            18.187294073
    4 120 3.208 107440 0            16.713171684
}\dataPoCLatencyDiemBFT

\pgfplotstableread{
    Id miners	batch	tput	latency	difficulty
    1 120 120000	73294.01853	671.30725 9
    2 120 160000    86330.93525 306.5873286 9
    3 120 200000    101010.10101 198.9935636 9
    4 120 240000    100289.29605 197.482 9
    5 120 280000	103251.76	247.6701 9
}\dataPoCScaleBatchDifficultyOne

\pgfplotstableread{
    Id miners	batch	tput	latency	difficulty
    1 120	120000	105263.1579	62.2502	8
    2 120	160000	102564.1026	74.0707	8
    3 120	200000	102040.8163	92.5225	8
    4 120	240000	101265.8228	116.111	8
    5 120   280000	102065.61361 131.6816667 8
}\dataPoCScaleBatchDifficultyTwo

\pgfplotstableread{
    Id miners	batch	tput	latency	difficulty
    2 120	64	12455.97776	24.06950	 7
    3 120	128	12439.26142	48.18780	7
}\dataPoCScaleBatchDifficultyThree

\pgfplotstableread{
    Id miners	batch	tput	latency	difficulty
    2 120	64	12305.97776	22.06950	6
    3 120	128	12463.48588	47.78900	6
}\dataPoCScaleBatchDifficultyFour

\pgfplotstableread{
    Id time tput	latency
    1 5	    59035.75252	330.743
    2 10	16540.64788	1489.63
    3 15	120019.6467	1411.79
    4 20	47140.32638	1673.91
    5 25	87309.61701	1691.62
    6 30	60819.8664	1829.74
    7 35	62134.75174	1969.32
    8 40	111832.7304	1899.7
}\dataPoCNoSolFive

\pgfplotstableread{
    Id time tput	latency
    1 5	90007.31689	137.289
    2 10	19943.75595	1302.53
    3 15	108679.6259	1254.36
    4 20	72794.61436	1342.73
    5 25	98519.90266	1326.31
    6 30	54098.08307	1546.49
    7 35	95885.6123	1507.6
    8 40	81509.71645	1534.36
}\dataPoCNoSolSix

\pgfplotstableread{
    Id time tput	latency
    1 5	88339.64727	186.483
    2 10	21506.40046	1326.93
    3 15	152453.2339	1179.05
    4 20	110078.3317	1117.84
    5 25	93125.65072	1119.79
    6 30	159342.3078	822.762
    7 35	143340.1113	664.353
    8 40	138251.6694	496.988
}\dataPoCNoSolSeven

\pgfplotstableread{
    Id time tput	latency
    1 5	90147.46774	147.631
    2 10	25418.06522	1116.29
    3 15	283554.1097	820.198
    4 20	223190.3725	571.565
    5 25	184159.2646	342.954
    6 30	145459.6762	147.576
    7 35	94605.34317	170.051
    8 40	85733.97391	84.8075
}\dataPoCNoSolEight

\newcommand{\evalShiftOnceLarge}[4]{\begin{tikzpicture}[plot]
        \begin{axis}[xlabel={Block Number},ylabel={#2}, title={(#3)},
                xlabel style={text width=5.8cm, align=center},
                xticklabels={5, 10, 15, 20, 25, 30, 35, 40},
                width=200,
                height=120,
                legend style={font=\footnotesize},
                legend pos={north east},
            ]
            \addlegendentry{1 malicious}
            \addplot[color=red, mark=square] table[x={Id},y={#1},  color=red, mark=square] {\dataPoCLagShiftOneFour};
            \addlegendentry{No Solution}
            \addplot[color=blue, mark=triangle] table[x={Id},y={#1},  color=blue, mark=traingle] {\dataPoCNoSolEight};
        \end{axis}
    \end{tikzpicture}}
\newcommand{\EvalLargeShiftTput}{\evalShiftOnceLarge{tput}{\axistputblock}{c}{}}
\newcommand{\EvalLargeShiftLatency}{\evalShiftOnceLarge{latency}{\axislat}{d}{}}

\pgfplotstableread{
    Id time tput	latency
    1 5	60459.72608	28.4173
    2 10	37176.20824	103.536
    3 15	105824.0924	93.9726
    4 20	101687.7937	89.9053
    5 25	98688.15805	78.6536
    6 30	122398.901	60.2537
    7 35	45145.16971	110.719
    8 40	90295.11695	107.701
}\dataPoCLagShiftTen

\pgfplotstableread{
    Id time tput	latency
    1 5	83364.46765	12.2827
    2 10	42465.7001	83.5639
    3 15	78427.86636	93.6613
    4 20	122662.0611	74.1922
    5 25	97176.06359	62.9398
    6 30	151306.0359	30.4943
    7 35	112478.5588	14.1932
    8 40	87756.23541	14.3924
}\dataPoCLagShiftOneTwo

\pgfplotstableread{
    Id time tput	latency
    1 5	72182.58482	31.8439
    2 10	55387.72197	75.9155
    3 15	140661.3899	48.844
    4 20	149113.3931	10.2843
    5 25	81220.91276	16.2427
    6 30	94964.19782	8.56475
    7 35	84989.92869	12.8622
    8 40	84339.07682	17.6567
}\dataPoCLagShiftOneFour

\pgfplotstableread{
    Id time tput	latency
    1 5	91408.84801	5.57552
    2 10	25286.95958	117.708
    3 15	127902.5894	102.671
    4 20	54626.27163	133.622
    5 25	125750.8505 118.425
    6 30	73889.39164	98.7602
    7 35	124994.6487	76.8462
    8 40	75772.25663 83.4589
}\dataPoCLagShiftEight

\pgfplotstableread{
    Id time tput	latency
    1 5	56722.78441	27.0355
    2 10	28037.41724	98.2443
    3 15	74092.36849	105.854
    4 20	52168.75968	131.041
    5 25	61556.68232	144.365
    6 30	86192.58066	143.788
    7 35	73731.92155	149.471
    8 40	89784.71421	147.146
}\dataPoCLagShiftSix

\newcommand{\evalresult}[8]{\begin{tikzpicture}[plot]
        \begin{axis}[xlabel={#4},ylabel={#5}, title={#6},
                xticklabels={#8},
                legend pos=#7,
                legend style={font=\footnotesize},
                width=200,
                height=120,
            ]

            \addlegendentry{B=160k}
            \addplot[color=blue, mark=triangle] table[x={Id},y={#3},  color=blue, mark=traingle] {\dataPoCScaleMinerSFN};
            \addlegendentry{B=200k}
            \addplot[color=red, mark=square] table[x={Id},y={#3},  color=red, mark=square] {\dataPoCScaleMinerOTEN};
            \addlegendentry{B=240k}
            \addplot[color=green, mark=pentagon] table[x={Id},y={#3},  color=black, mark=pentagon] {\dataPoCScaleMineSix};
        \end{axis}
    \end{tikzpicture}}

\newcommand{\EvalMinerParam}[4]{\evalresult{\dataPoCScaleMinerSFN}{\dataPoCScaleMinerOTEN}{#1}{\axisnodes}{#2}{{(#3)}}{#4}{64,96,128}}
\newcommand{\EvalMinerTput}[1]{\EvalMinerParam{tput}{\axistputblock}{#1}{south west}}
\newcommand{\EvalMinerLatency}[1]{\EvalMinerParam{latency}{\axislat}{#1}{south west}}

\pgfplotstableread{
    Id miners	batch	tput	latency	difficulty
    1 64    160000 62717.82194 607.8313864	9
    2 96    160000	81528.66242	494.8340455 9
    3 120	160000 86330.93525 306.5873286 9
}\dataPoCScaleMinerSFN

\pgfplotstableread{
    Id miners	batch	tput	latency	difficulty
    1 64    200000	98147.82780 496.9644667 9
    2 96    200000	103708.24524 207.6503636 9
    3 120	200000  101010.10101 198.9935636 9
}\dataPoCScaleMinerOTEN

\pgfplotstableread{
    Id miners	batch	tput	latency	difficulty
    1 64	240000 98749.96714 450.1665263 9
    2 96    240000 100313.47962 261.9214 9
    3 120   240000 107126.39761 197.482 9
}\dataPoCScaleMineSix

\newcommand{\PbftBatchtput}{\begin{tikzpicture}[plot]
        \begin{axis}[xlabel={Batch Size},ylabel={Throughput (txn/s)}, title={(a)},
                xticklabels={1, 50, 100, 150, 200, 500, 1000},
                ytick={1000, 40000, 100000,130000},
                yticklabels={1k, 40k, 100k, 130k},
                ymax=130000,
                width=200,
                height=120,
            ]
            \addplot[color=blue, mark=square] table[x={Id},y={tput},  color=blue, mark=square] {\dataPoCScaleBatch};
        \end{axis}
    \end{tikzpicture}}

\newcommand{\PbftBatchlat}{\begin{tikzpicture}[plot]
        \begin{axis}[xlabel={Batch Size},ylabel={Latency (s)}, title={(b)},
                xticklabels={1, 50, 100, 150, 200, 500, 1000},
                width=200,
                height=120,
            ]
            \addplot[color=blue, mark=square] table[x={Id},y={latency},  color=blue, mark=square] {\dataPoCScaleBatch};
        \end{axis}
    \end{tikzpicture}}

\pgfplotstableread{
    Id batch	tput	latency
    1 1    1003   0.131301
    2 50   51490  0.356074
    3 100  107126 0.997482
    4 150  121126 1.27315
    5 200  123480 3.1678
    6 500  123700 7.38218
    7 1000 122000 14.5174
}\dataPoCScaleBatch

\pgfplotstableread{
    Id  data ida_encode ida_decode tree_encode  tree_verify old_size  new_size
    1 100*64k 1646.123   411.672   059.1427   034.1054   6408.873   160.510
    2 50*64k  831.662   212.932   029.8309   017.6719   3204.373    80.398
    3 25*64k  466.318   102.765   015.1731   009.73672  1602.123    40.342
}\dataIDA

\pgfplotstableread{
    Id Nodes Throughput Latency
    1 16	723	0.1
    2 32	393	0.2
    3 64	190	0.29
    4 120   74  0.3
}\dataHyperLedgerThreeOrderer

\pgfplotstableread{
    Id Nodes Throughput Latency loglatency
    1 16 700	0.1	-3.321928095
    2 32 367	0.2	-2.321928095
    3 64 173	0.4 -1.321928095
    5 120 74    0.9 -0.1520030934
}\dataHyperLedgerSameOrderer
\pgfplotstableread{
    Id Nodes Throughput Latency BlockSize Difficulty BlockPS loglatency
    1 16	80	170	23	38000000	0.1  7.409390936
    2 32	90	163	30	45000000    0.1	 7.348728154
    3 64	87	193	34	43000000    0.1  7.592457037
    5 120   85   176 36  0 0.1   6.906890596
}\dataEthereum

\pgfplotstableread{
    Id Nodes Throughput Latency loglatency
    1 16	230     1.159231313  0.2131684704
    2 32    200     1.244670489  0.3157638573
    3 64    140     1.424466586  0.5104217805
    4 96    65      2.332754374  1.222034408
}\dataCardano

\pgfplotstableread{
    Id Nodes Throughput Latency loglatency
    1 16	336     6.05  2.596935142
    2 32    318     5.21  2.381283373
    3 64    264      4.25 2.087462841
    5 120   189      4.4  2.137503524
}\dataAvalanche

\pgfplotstableread{
    Id Nodes Throughput Latency loglatency logtps
    1 16	9912     0.104  0               13.274960472
    2 32    9560     0.670  0               13.222794903
    3 64    2858     1.231  0.29983         11.480790201
    4 120   5499     3.856  1.947           12.424953571
}\dataDiemBFTNoMempool

\pgfplotstableread{
    Id name Throughput Latency TCost
    1  NexResPbft  121126 1.27315 0.005944223371
    2  PoC  101010.101 198.9935636 0.007128 
    3  DiemBFT 107440 3.208 0.006701414743
    4  hyperledger 74 0.9 9.72972973
    5  eth 7.1 176 101.4084507
    6  cardano 65 2.332754374 11.07692308
    7  avalanche 189 4.4 3.80952381
}\dataTCost

\pgfplotstableread{
Id Workload Throughput Latency
1	168	150.7	11.262
1	852	568.9	17.031
1	1000	562.5	19.541
1	3438	511.4	20.109
1	13303	633.3	21.859
1	37976	116.8	28.590
}\dataDiabloAlgorand

\pgfplotstableread{
Id Workload Throughput Latency
1	168	139.5	17.91
1	852	413.2	19.769
1	1000	413.3	19.623
1	3483	419.9	27.943
1	13033	570.1	29.721
1	37976	82.2	32.385
}\dataDiabloAlgorandPoC

\pgfplotstableread{
Id Workload Throughput Latency
1	168	167.5	0.402
1	852	836.9	1.746
1	1000	978.5	2.733
1	3438	849.4	48.648
1	13303	903.8	62.382
1	37976	33.3	6.286
}\dataDiabloDiem

\pgfplotstableread{
Id Workload Throughput Latency
1	168	164.1	6.021
1	852	769.9	11.201
1	1000	723.3	20.799
1	3483	768.3	68.067
1	13303	0	0
1	37976	0	0
}\dataDiabloDiemPoC

\pgfplotstableread{
Id Workload Throughput Latency
1	168	164.6	4.612
1	852	372.3	6.967
1	1000	238.2	8.42
1	3483	48.1	9.781
1	13033	32.2	12.249
1	37976	17.3	21.406
}\dataDiabloQuorum

\pgfplotstableread{
Id Workload Throughput Latency
1	168	160.6	7.687
1	852	244.2	7.583
1	1000	204.1	14.126
1	3483	42.4	15.212
1	13033	24.2	20.485
1	37976	1.4	27.684
}\dataDiabloQuorumPoC

\pgfplotstableread{
Id Workload Throughput Latency
1	168	7.6	59.752
1	852	7.6	60.125
1	1000	7.6	62.149
1	3438	7.6	60.961
1	13303	8.7	60.645
1	37976	4.4	33.507
}\dataDiabloEthereum

\pgfplotstableread{
Id Workload Throughput Latency
1	168	7.6	60.246
1	852	7.8	76.948
1	1000	7.6	61.602
1	3483	8.7	79.413
1	13033	8.7	67.478
1	37976	3.3	51.9
}\dataDiabloEthereumPoC

\pgfplotstableread{
Id Workload Throughput Latency
1	168	163.7	0.543
1	852	847.6	0.432
1	1000	994.2	0.308
1	3438	3422.8	0.192
1	13303	13152.4	0.111
1	37976	37021.5	0.093
}\dataDiabloNexres

\pgfplotstableread{
Id Workload Throughput Latency
1	168	151.3	11.896
1	852	832	16.894
1	1000	817.4	18.612
1	3483	3106.8	13.057
1	13033	12341.1	7.991
1	37976	35759.3	10.227
}\dataDiabloNexresPoC

\pgfplotsset{select coords index/.style 2 args={
    x filter/.code={
        \ifnum\coordindex<#1\fi
        \ifnum\coordindex>#1\fi
    }
}}

\definecolor{colAA}{RGB}{255,87,51}
\definecolor{colAB}{RGB}{255,116,51}
\definecolor{colBA}{RGB}{53,142,250}
\definecolor{colBB}{RGB}{147,199,251}
\definecolor{colCA}{RGB}{92,208,48}
\definecolor{colCB}{RGB}{167,252,105}
\definecolor{colDA}{RGB}{214,217,60}
\definecolor{colDB}{RGB}{242,245,52}
\definecolor{colEA}{RGB}{242,84,228}
\definecolor{colEB}{RGB}{223,128,216}
\definecolor{colFA}{RGB}{159,28,249}
\definecolor{colGA}{RGB}{28,71,249}

\newcommand{\DiabloTput}[5]{\begin{tikzpicture}[plot]
        \begin{axis}[ybar, title={#3},ylabel={#5},
                xticklabels={#2},
                width=180,
                height=120,
                bar width={0.35cm},
                legend to name={mainlegend},
                legend columns=10,
                legend style={font=\footnotesize},
            ]

            \addplot[fill=colAA, pattern=horizontal lines, select coords index={#1}] table[y={#4}] {\dataDiabloAlgorand};
            \addplot[fill=colAB, pattern=horizontal lines, select coords index={#1}] table[y={#4}] {\dataDiabloAlgorandPoC};
            \addplot[fill=colBA, pattern=north east lines, select coords index={#1}] table[y={#4}] {\dataDiabloDiem};
            \addplot[fill=colBB, pattern=north east lines, select coords index={#1}] table[y={#4}] {\dataDiabloDiemPoC};
            \addplot[fill=colDA, pattern=vertical lines, select coords index={#1}] table[y={#4}] {\dataDiabloEthereum};
            \addplot[fill=colDB, pattern=vertical lines, select coords index={#1}] table[y={#4}] {\dataDiabloEthereumPoC};
            \addplot[fill=colEA, pattern=north west lines, select coords index={#1}] table[y={#4}] {\dataDiabloQuorum};
            \addplot[fill=colEB, pattern=north west lines, select coords index={#1}] table[y={#4}] {\dataDiabloQuorumPoC};
            \addplot[fill=black, select coords index={#1}] table[y={#4}] {\dataDiabloNexres};
            \addplot[fill=black, select coords index={#1}] table[y={#4}] {\dataDiabloNexresPoC};
            \legend{\Algorand{},\Algorand{}+\PoC{},\Diem{},\Diem{}+\PoC{}, \Ethereum{}, \Ethereum{}+\PoC{}, \Quorum{},\Quorum{}+\PoC{}, \ResDB{}, \ResDB{}+\PoC{}};

        \end{axis}
    \end{tikzpicture}}

\newcommand{\DiabloNasTput}{\DiabloTput{0}{}{\text{workload = 168 TPS}}{Throughput}{\axistput}}
\newcommand{\DiabloNasLat}{\DiabloTput{0}{NASDAQ}{}{Latency}{\axislat}}

\newcommand{\DiabloVisaTput}{\DiabloTput{2}{}{\text{workload = 1000 TPS}}{Throughput}{}}
\newcommand{\DiabloVisaLat}{\DiabloTput{2}{Visa}{}{Latency}{}}

\newcommand{\DiabloFifaTput}{\DiabloTput{3}{}{\text{workload = 3483 TPS}}{Throughput}{\axistput}}
\newcommand{\DiabloFifaLat}{\DiabloTput{3}{Fifa}{}{Latency}{\axislat}}

\newcommand{\DiabloDotaTput}{\DiabloTput{4}{}{\text{workload = 13303 TPS}}{Throughput}{}}
\newcommand{\DiabloDotaLat}{\DiabloTput{4}{Dota 2}{}{Latency}{}}

\newcommand{\DiabloYoutubeTput}{\DiabloTput{5}{}{\text{workload = 37976 TPS}}{Throughput}{}}
\newcommand{\DiabloYoutubeLat}{\DiabloTput{5}{Youtube}{}{Latency}{}}

\pgfplotstableread{
Id State(wm/bs) Throughput Latency LogTPS
1	1/1/1	10	0.0296544	3.321928095
2	128/1/1	677	0.171783	9.403012024
3	1024/1/1	866	1.3878	9.758223215
4	1024/1/8	2878	0.408802	11.49085088
5	1024/100/8	256380	0.455462	17.9679242
6	1024/2008	461560	0.475212	18.81615868
7	1024/200/16	1026120	0.197712	19.96876803
}\dataNexresSplitSmall

\pgfplotstableread{
Id State(wm/bs) Throughput Latency LogTPS
1	1/1/1	10	0.0354019	3.321928095
2	128/1/1	250	0.592593	7.965784285
3	1024/1/1	259	5.06772	8.016808288
4	1024/1/8	872	1.47411	9.768184325
5	1024/100/8	82720	1.59013	16.33594857
6	1024/2008	158680	1.64252	17.27576078
7	1024/200/16	305320	0.803505	18.21996257
}\dataNexresSplitLarge

\pgfplotstableread{
Id miners c1 c2 c4 c6 c8
1 1	10.00	20.00	40.00	60.00	80.00
2 2	1.00	4.00	16.00	36.00	64.00
3 4	0.01	0.16	2.56	12.96	40.96
4 8	0.00	0.00	0.07	1.68	16.78
5 16	0.00	0.00	0.00	0.03	2.81
6 32	0.00	0.00	0.00	0.00	0.08
7 64	0.00	0.00	0.00	0.00	0.00
8 128	0.00	0.00	0.00	0.00	0.00
}\dataConsecutiveBlock

\pgfplotstableread{
Id c0 c1 c4 c8 c16 c32 c64 c128 c256 c512 c1024 c2048 c4096
1	4.584962501	3.584962501	1.584962501	0.5849625007	-0.4150374993	-1.415037499	-2.415037499	-3.415037499	-4.415037499	-5.415037499	-6.415037499	-7.415037499	-8.415037499
2	5.584962501	4.584962501	2.584962501	1.584962501	0.5849625007	-0.4150374993	-1.415037499	-2.415037499	-3.415037499	-4.415037499	-5.415037499	-6.415037499	-7.415037499
3	6.169925001	5.169925001	3.169925001	2.169925001	1.169925001	0.1699250014	-0.8300749986	-1.830074999	-2.830074999	-3.830074999	-4.830074999	-5.830074999	-6.830074999
4	6.584962501	5.584962501	3.584962501	2.584962501	1.584962501	0.5849625007	-0.4150374993	-1.415037499	-2.415037499	-3.415037499	-4.415037499	-5.415037499	-6.415037499
5	6.906890596	5.906890596	3.906890596	2.906890596	1.906890596	0.9068905956	-0.09310940439	-1.093109404	-2.093109404	-3.093109404	-4.093109404	-5.093109404	-6.093109404
}\dataAttackYear

\makeatletter
\newcommand{\linebreakand}{%
  \end{@IEEEauthorhalign}
  \hfill\mbox{}\par
  \mbox{}\hfill\begin{@IEEEauthorhalign}
}
\makeatother

\begin{document}

\title{
Securing Consensus from Long-Range Attacks through Collaboration\\
\thanks{
  This work is partially funded by NSF Award Number 2245373.
}
}


\author{
\IEEEauthorblockN{Junchao Chen}
\IEEEauthorblockA{\textit{Exploratory Systems Lab} \\
\textit{University of California, Davis}\\
California, United States\\
jucchen@ucdavis.edu}
\and
\IEEEauthorblockN{Suyash Gupta}
\IEEEauthorblockA{\textit{University of Oregon} \\
Oregon, United States\\
suyash@uoregon.edu}
\and
\IEEEauthorblockN{Alberto Sonnino}
\IEEEauthorblockA{\textit{Mysten Labs} \\
\textit{University College London (UCL)}\\
London, England\\
alberto@mystenlabs.com}
\linebreakand 
\IEEEauthorblockN{Lefteris Kokoris-Kogias}
\IEEEauthorblockA{\textit{Mysten Labs} \\
London, England\\
lefteris@mystnelabs.com}
\and
\IEEEauthorblockN{Mohammad Sadoghi}
\IEEEauthorblockA{\textit{Exploratory Systems Lab} \\
\textit{University of California, Davis}\\
California, United States\\
msadoghi@ucdavis.edu}
}


\maketitle

\begin{abstract}
  Decentralized systems built around blockchain technology promise clients 
an immutable ledger. 
They add a transaction to the ledger after it undergoes consensus among the replicas 
that run a Proof-of-Stake (\PoS{}) or Byzantine Fault-Tolerant (\BFT{}) consensus protocol.
Unfortunately, these protocols face a long-range attack where an adversary having access 
to the private keys of the replicas can rewrite the ledger.
An existing solution to this problem forces each committed block from these protocols to undergo another consensus, Proof-of-Work (\PoW{}) consensus;
\PoW{} protocol wastes computational resources as miners compete to solve complex puzzles.
In this paper, we present the design of our Power-of-Collaboration (\PoC{}) protocol, which guards existing
\PoS{}/\BFT{} blockchains against long-range attacks and requires miners to collaborate rather than compete.
\PoC{} guarantees fairness and accountability and only marginally degrades the throughput of the underlying system.

\end{abstract}


\section{Introduction}
\label{s:intro}
Decentralized systems  built using blockchain technology promise their clients an immutable and verifiable ledger~\cite{bitcoin,ether,algorand}. 
These systems receive client transactions and use state machine replication to 
add these transactions to the ledger.
As these systems are often composed of untrusting nodes, some of which are malicious or Byzantine, 
establishing state machine replication requires these systems to run a {\em consensus} protocol that can handle malicious attacks~\cite{pbftj,hotstuff}.
The two most widely adopted categories of these consensus protocols are:
{Proof-of-Stake} (\PoS{}) protocols~\cite{algorand,ppcoin} and traditional Byzantine Fault-Tolerant (\BFT{}) protocols~\cite{pbftj,hotstuff}.
 
Stake-oriented consensus protocols, such as Proof-of-Stake (\PoS{}) protocols, 
use a probabilistic distribution to decide which node gets to add a new block of client transactions to the ledger; 
often, the nodes with a higher stake (or wealth) have a higher probability of proposing a new block~\cite{ppcoin}. 
Communication-oriented protocols, such as traditional \BFT{} protocols, 
give each node an equal opportunity (a vote) to add an entry to the ledger; 
agreement on the next block is reached through successive rounds of vote exchange~\cite{pbftj,zyzzyva}. 
Some systems combine \PoS{} and \BFT{} to yield efficient consensus~\cite{algorand}.
Despite these differences, these protocols follow the same design: 
each block added to the ledger includes the {\em digital signatures} of a quorum of participants to prove
that a quorum agreed to update the ledger.

Unfortunately, any decentralized system that employs these protocols suffers from a well-known attack: 
a {\em long-range attack} where an adversary attempts to create an 
alternate ledger and targets clients (or new participants) that cannot distinguish between the original ledger 
and the adversarial ledger~\cite{winkle,bft-longrange,sperax,bitcoin-pos}.
An adversary can launch a long-range attack on systems running \PoS/\BFT{} consensus protocols due to the following reason:
In \PoS/\BFT{} protocols, it is {\em computationally inexpensive} for nodes to add a new block to the ledger.
An adversary with {\em access to the private keys} of the honest nodes can use these keys to create an alternate ledger;
following are the two ways to access the private keys of others: 

\begin{enumerate}
\item {\em Stealing.}
An adversary can attempt to steal the keys of the nodes;
stealing private keys is a widespread attack, and such attacks have resulted in losses of up to \$200 million~\cite{cointelegraph-keyattack,bein-keyattack,private-key-breaches}.

\item {\em Bribing.}
An adversary can bribe honest nodes to sell their private keys, especially nodes that once participated in the system and no longer have any stake in it.
This bribery attack is feasible because decentralized systems expect to run for years and 
cannot guarantee that the original set of participants will always run the system.
Based on the Tragedy of the Commons~\cite{proof-of-activity}, rational participants will opt to earn further incentives by selling their keys.
\end{enumerate}

Once an adversary has access to these keys, it can use them to fork the original ledger at a specific block number 
and create an adversarial ledger with alternate blocks.
\Modify{If it has stolen over 50\% of the private keys, the adversary can control the whole train.} 
Or {\em the adversary waits for existing nodes to leave and new nodes to join the system}.
Once it discovers a new node interested to join the system, 
the adversary presents its adversarial ledger as the authentic ledger 
to the new node, which unfortunately {\em cannot distinguish between the two}.
It is hard for new nodes to distinguish between the ledgers despite the existence of the honest nodes in the system who have access to the original ledger because the adversary has used the private keys of such honest nodes to forge their identities.

A naive solution to this problem is to ensure that honest nodes never leave the system and that no new node can join it. 
However, this solution is impossible to implement in a decentralized setting as nodes frequently leave/join the system~\cite{ethereum-exit,bitcoin-entry}.
Other prior attempts include:
(1) Using key-evolving cryptographic techniques and increasing the number 
of keys an adversary needs to compromise~\cite{winkle,key-evolving}, which typically delays the imminent long-range attack.
(2) Creating state checkpoints and storing them at all the nodes, assuming that an 
adversary can only compromise the keys of at most one-third of nodes~\cite{gasper,bitcoin-pos} 
(3) Periodically appending the ledger state to the Bitcoin blockchain~\cite{tas2022babylon,azouvi2022pikachu}.
Indeed, the third direction can guard existing decentralized systems against long-range attacks. 
For an adversary to present an adversarial chain to the new nodes, it also needs to rewrite the Bitcoin ledger, which is computationally infeasible.
Bitcoin employs the \PoW{} consensus protocol,
which follows a {\em computation-oriented} model as it requires all the nodes to compete toward solving a complex puzzle. 
Whichever node solves the puzzle first adds a new block and receives a reward as compensation for its efforts. 
As \PoW{} nodes constantly compete with each other, \PoW-based solutions lead to the wastage of computational resources, as there is only one winner.~\cite{badcoin}.

The challenges existing solutions face while eliminating long-range attacks make us conclude that any solution for long-range attacks should:
(1) not rely on the long-term safe-keeping of private keys,
(2) reduce wastage of computational resources, and
(3) be computationally expensive for an attacker to rewrite the ledger.

In this paper, we introduce {\em Power-of-Collaboration} (\PoC{}) protocol, which, when 
appended to decentralized systems running \PoS{}/\BFT{} consensus, helps to meet the aforementioned goals.
\PoC{} is noninvasive as it works on the output of underlying \PoS{}/\BFT{} consensus protocol and
has minimal impact on the performance of existing decentralized systems.
\PoC{} advocates for {\em collaborative mining}, which, like \PoW{}, requires miners to solve a compute-intensive puzzle, 
but all the miners are now working together (instead of competing) to solve the same puzzle.

The most closely related work, Bitcoin's {\em centralized mining pools}, also attempts to reduce the costs associated with mining~\cite{smartpool,p2pool,poolparty}. 
As the name indicates, these mining pools are centralized and managed by an organization. 
The organization sets the rules for the mining pool, decides which node should receive a reward and how much reward, and controls which node can participate in the pool.
Not only is the existence of centralized mining pools against the ethos of a decentralized system, 
but the managing organization charges fees for management without spending any computational resources.
Further, attempts to create a decentralized mining pool have been unsuccessful due to nodes not doing designated tasks and lack of accountability:
the last block added by any decentralized mining pool in Bitcoin was in 2019~\cite{poolparty,smartpool}.

\PoC{}, in essence, functions as a single decentralized mining pool where all the nodes collaborate to find a solution for the compute-intensive puzzle. 
Like \PoW{}, nodes are still spending their computational resources to find the nonce, 
which makes it computationally expensive for the adversary to create an adversarial ledger. 
However, we need to ensure that, like centralized mining pools, we reduce the wastage of computational resources while also guaranteeing decentralization and {\em fairness}. 
We do so by splitting the compute-intensive puzzle into a set of unique sub-problems, and 
each node works on a unique subset of these sub-problems; the solution to the original compute-intensive problem is present in these subsets. 
We also need to provide {\em accountability} and deter malicious nodes from not doing work, as it can delay the discovery of the solution. 
\PoC{} does so through our slice-shifting protocol, which identifies and penalizes a malicious miner and transfers its work to honest miners.
\Modify{To allow participants join and leave the mining group, 
\PoC{} also provides a reconfiguration solution to reschedule the problem assignment.}

To show that \PoC{} is effective in practice,
we append it to several decentralized systems.
In our first set of experiments, we append \PoC{} to Apache's \ResDB{} (Incubating)~\cite{apache-resdb} as
it provides access to an open-source permissioned blockchain platform and 
an optimized implementation of \PBFT{}, a \BFT{} consensus protocol.
\ResDB's \pbft{} implementation adds approximately $1000$ blocks per second on a system of $128$ replicas, and 
our experiments illustrate that \PoC{} can sustain this throughput on a system of $128$ miners and requires $29\times$ less mining time than 
Bitcoin's \PoW{} protocol.
In our final set of experiments, we use the \Diablo{}~\cite{diablo} benchmarking framework 
to append \PoC{} to {\em four} popular blockchain systems, namely \Diem{}~\cite{diembft},
\Algorand{}~\cite{algorand}, \Ethereum~\cite{ether}, and \Quorum{}~\cite{quorum}.
Our results illustrate that \PoC{} has a minimal impact ($\approx10\%$) on the throughput of these blockchains.
Next, we list our contributions.
\begin{itemize}
    \item We present the Power-of-Collaboration (\PoC{}) protocol, which, 
	when appended to existing decentralized systems running \PoS{} and \BFT{} protocols, makes it computationally-expensive 
	for an adversary to launch a long-range attack.

    \item \PoC{} introduces the notion of collaborative mining, which divides the mining task among all the miners.

    \item \PoC{} advocates fairness and accountability: 
	rewards are distributed among the miners in proportion to their share of work, and 
	Byzantine behavior is quickly detected and penalized through the slice-shifting mechanism.
\end{itemize}

{\em Outline.}
\fulltext{
In \S\ref{sec:preliminaries}, we present the system model. 
}
In \S\ref{s:back}, we discuss various types of consensus protocols, 
pooled mining, and long-range attacks.
In \S\ref{s:dual} and \S\ref{s:attacks}, 
we present the design of our \PoC{} protocol and discuss the impact of malicious attacks.
In \S\ref{s:discussion}, we discuss the impact of mining difficulty and
reconfiguration on POC and present the correctness proof and security analysis.

\section{Preliminaries} \label{sec:preliminaries}

\Modify{{\bf System Model.}}
We adopt the standard communication and failure model adopted by most consensus protocols~\cite{pbftj,zyzzyva,sbft}.
We assume the existence of a decentralized system 
\fulltext{
$\Service$ of the form 
}
$\Service = \{\Replicas{}, \Clients{} \}$. The set $\Replicas$
consists of $\n{\Replicas{}}$ replicas (or stakeholders in case of \PoS{} protocols)
of which at most $\f{\Replicas{}}$ can behave arbitrarily \Modify{and $\n{\Replicas{}} \ge 3\f{\Replicas{}} + 1$}. 
\fulltext{
In a typical decentralized system, these replicas store the state and participate in consensus. 
}
The remaining
$\n{\Replicas}-\f{\Replicas}$ are honest: they follow the protocol and remain live.
We also assume the existence of a finite set of clients $\Clients$, of which
arbitrarily many can be malicious.

\Modify{We denote the set of miners as $\Miners$ and the total number of miners as $\n{\Miners{}}$,
of which at most $\f{\Miners{}}$ can act maliciously. 
Following the same setting of Bitcoin~\cite{bitcoin} and Ethereum~\cite{ethereum-stake-withdrawl},
the system can tolerant half of the adversary miners ($\n{\Miners} \ge 2\f{\Miners}+1$).}

{\bf Miner Staking.}
\fulltext{
Unlike \PoW-based systems, where any node can start mining without informing everyone about its existence, 
}
We make similar assumptions as most \PoS-based systems~\cite{algorand,cardano}:
we require knowledge of the total number of miners participating in the mining process.
Like all the \PoS{} systems, which require any node wishing to become a stakeholder to {\em stake} some of its resources, 
we need each miner wishing to participate in the \PoC{} mining to {\em stake} its resources.
This staking determines how much work a miner must perform
(more on this in \S~\ref{ss:collaboration} and~\ref{ss:staking}).


{\bf Authenticated communication.} Replicas/miners employ
\fulltext{
standard cryptographic primitives such as 
}
\MAC{} and digital signatures (\DS{})
to sign messages and accept only {\em well-formed} messages.
We use a \emph{collision-resistant}
hash function $\Hash{\cdot}$ to map an arbitrary value $v$ to a constant-size digest~\cite{cryptobook}. 
\fulltext{
Each replica/miner only accepts a message if it is {\em well-formed}.
}

{\bf Standard Adversary model.}
Almost all the prior systems assume this adversarial model:
the adversary can corrupt at most $\f{\Replicas}$ replicas and $\f{\Miners}$ miners,
delay, and reorder messages~\cite{pbftj,algorand,bitcoinng}.
Byzantine replicas can perform any attack permitted by the underlying \PoS/\BFT{} protocol.
Byzantine miners can avoid participating in the mining protocol and issue invalid solutions for the puzzle.

    {\bf Advanced Adversary model.}
Additionally, we assume that the adversary can somehow access the private keys of all the replicas/miners.
Using these private keys, the adversary can attempt a long-range attack to overwrite the \PoS/\BFT{} ledger.

We assume that the underlying consensus protocol states a mechanism 
for replicas to join or leave $\Service$, and this knowledge is percolated
to all the existing members of $\Service$.
In \S~\ref{ss:reconfiguration}, we discuss how \PoC{} miners can leave or join the system.
\PoC{} offers {\bf Sybil resistance} like existing \PoS{} systems; 
each miner must stake its resources before the start of the mining, 
which it cannot arbitrarily cash out.

{\bf Anonymity.}
Like existing \PoW{} systems, we assume {\em pseudo-anonymity}
for miners in set $\Miners$; they are identified only through their {\em public keys}, which they may hold many.
Similarly, like any \PoS{} system knows the total number of stakeholders,
we know $\n{\Miners}$, which allows us to assign a unique identifier to each miner in the range  $[0, \n{\Miners}]$.
%
Each replica/stakeholder is also assigned an identifier in the range of $[0, \abs{\Replicas}]$.

Finally, we expect the underlying \PoS/\BFT{} protocol to provide the following standard guarantees:
\begin{description}[\IEEEsetlabelwidth{Liveness}\IEEEusemathlabelsep]
    \item[\bf Safety.]
        If two honest replicas $\Replica{1}$ and $\Replica{2}$ order transactions
        $\Transaction{}$ and $\Transaction'$ at sequence numbers $k$, then $\Transaction{} = \Transaction'$.

    \item[\bf Liveness.]
        If a client sends a transaction $\Transaction{}$, then it will eventually
        receive a response for $\Transaction{}$.
\end{description}

\section{Background}
\label{s:back}
We begin by presenting the necessary conceptual background.

\subsection{\PoS{} and \BFT{} Consensus}
\PoS{} consensus protocols allow each node to add the next block to the blockchain in proportion to its invested stake~\cite{algorand,cardano}.
Often, the stake is equivalent to a monetary token or currency. 
Once a stakeholder
proposes the next block, all the other nodes also sign this block, which acts like an agreement
among the nodes.
Similarly, \BFT{} protocols~\cite{pbftj,hotstuff} designate in each round a
replica as a leader, which proposes a block.
Following this, all the replicas work through multiple rounds of message exchange to ensure that the proposed block has the
support of a quorum of honest replicas.

\subsection{Long-range attack on \PoS{} and \BFT{}}
A known attack that affects both \PoS~\cite{tas2022babylon,azouvi2022pikachu} and \BFT~\cite{bft-longrange,winkle} protocols
is the long-range attack, where an attacker attempts to create an alternate ledger.
As described in ~\S\ref{s:intro},
in \PoS/\BFT{} protocols, it is computationally inexpensive for nodes to add a new block to the ledger.
Thus, an adversary needs access to the private keys of the honest nodes, which it can do either through stealing or bribing.
Once an adversary has access to these keys, it can use them to fork the original ledger at a specific block number or height
and create an adversarial ledger with alternate blocks
(orthogonal, but in the past, blockchains have observed forks due to malicious attacks~\cite{dao-ethereum-split}).
As nodes of decentralized systems frequently leave/join the system~\cite{ethereum-exit,bitcoin-entry},
the adversary can use this opportunity to present its adversarial ledger as the authentic ledger
to a new node (or client), which unfortunately cannot distinguish between the two.
We illustrate this through the following example.

\begin{example} \label{ex:long-range}
	Assume that a decentralized system $\Service$ has the following \PoS{} blockchain ledger: $\block_1, \block_2, ... \block_k, ... \block_n$.
	Say malicious nodes get access to the private keys of all the honest nodes and
	decide to create an adversarial ledger, starting from the $k$-th block.
	Once it is the turn of malicious nodes to propose new blocks, they reveal the following adversarial ledger:
	$\block_1, \block_2, ... \block_k', ... , \block_n', \block_{n+1}'$.
	Any new node joining $\Service$ cannot distinguish between these two ledgers
	and will choose the longest chain. Similarly, some existing honest nodes, if bribed, may decide to forfeit
	their ledger and switch to the malicious ledger.
	Moreover, as time passes, with old nodes leaving the system and new nodes unable to distinguish,
	a hard-working adversary may be able to affirm the adversarial ledger as the original ledger.
\end{example}

In practice, there are a lot of examples where an adversary has successfully stolen the private keys of honest parties~\cite{cointelegraph-keyattack,bein-keyattack,private-key-breaches}.
We agree that stealing so many keys, primarily when the nodes are distributed is hard.
Hence, a rational attack is where the adversary bribes the honest validators who no longer have a stake in the system.
As these validators have nothing to lose, Tragedy of the Commons~\cite{proof-of-activity} suggests that these validators will sell
their private keys in return for some incentive.

A naive solution to this problem is to fix the set of nodes.
In that case, even if the adversary has access to the private keys of these nodes,
it cannot convince the honest nodes to switch to the adversarial ledger as, locally, each of them has a copy of the ledger.
Unfortunately, it is hard to prevent old nodes from leaving the system.
New nodes will eventually fill those spots, and the adversary needs to target only these new nodes.

The challenges make us conclude that any solution for long-range attacks should:
(1) not rely on the long-term safe-keeping of private keys, and
(2) be computationally expensive for an attacker to rewrite the ledger.

\subsection{Proof-of-Work Consensus}
\label{ss:pow}
We briefly look at the design of \PoW{} consensus protocol, which can help in preventing long-range attacks.

In the \PoW{} protocol, each miner $\Miner \in \Miners$ selects some client transactions
from the available pool of transactions and packs them in a block $\block$.
This block $\block$ also includes a header, which contains:
(i) hash of the previous block ($\prev$),
(ii) the digest of all transactions or {\em Merkle root} $M_\block$,
(iii) difficulty $D$, and 
(iv)  the nonce $\Nonce$, among other fields~\cite{smartpool,bitcoin-developer-reference}. 
As each miner decides which transactions to include in its block,
two miners may mine blocks with different transactions that {\em extend} the same previous block (with the hash $\prev$).  
Computing $M_\block$ of all transactions in the block requires a miner $\Miner$ to compute a
pairwise hash from leaves to the root.
The difficulty $D$, also termed as the difficulty of finding the nonce, informs the miner of the range of {\em desired hash}~\cite{garay2024bitcoin}.
Specifically, each miner continuously selects a random nonce $\Nonce$ till it satisfies the following equation:
\begin{equation}\label{eq:difficulty}
	\Hash{\prev ~||~ M_\block ~||~ \Nonce} ~<~ D
\end{equation}
When $\Miner{}$ discovers a {\em valid} nonce, it adds it to its block and broadcasts this block to all the miners.
When another miner $\Miner{'}$ receives a block with a valid nonce that extends the last block added to the ledger (with hash $\prev$), 
$\Miner{'}$ adds the received block to its ledger and starts building/mining the next block that extends the received block.
Note: once $\Miner{'}$ has added a block to the ledger, if in the future, $\Miner{'}$ receives any other block that includes $\prev$, it ignores that block.
Consequently, the miner who discovers the nonce earliest has the highest probability of adding a new block to the ledger as 
its block can reach a majority of miners the earliest.
Clearly, \PoW{} miners compete with each other in an attempt to find a valid nonce; 
\PoW{} consensus faces the following two challenges (among many others):
(1) All but one miner waste their computational resources and only the winner receives an incentive for finding the nonce.
(2) More than one miner can find a valid nonce, which can temporarily fork the ledger.
As there is no longer one ledger and instead multiple forks, 
\PoW{} protocols define a mechanism to trim all but one fork, leading to further wastage of computational resources.

One solution to reduce the probability of forks is to increase the hardness/difficulty of finding the nonce;
decentralized systems dynamically update the difficulty $D$ to fix the rate miners add new blocks to the ledger.
These systems want to ensure that miners spend at least a fixed amount of time searching for the valid nonce.
The value of $D$ is a system parameter and $D$ increases if miners are producing blocks at a faster rate than expected or
the probability of forks is high and decreases vice versa.


\fulltext{
\subsection{Centralized Pooled Mining}
Several decentralized systems, like Bitcoin, allow miners to work in groups to reduce the cost incurred by miners during \PoW{} consensus;
miners pool together their resources to increase their chances of finding a valid nonce~\cite{smartpool,p2pool}.
Arguably, almost all the active mining pools today are centralized;
they are run by an organization that manages the pool's functioning.

The pool controller creates the block for the pool miners to mine and determines a set of lower-difficulty sub-problems;
 assume that the expected difficulty for adding a block to the ledger is $D$, then a miner may need to find a nonce at difficulty $d << D$~\cite{garay2024bitcoin}.
If a miner finds a valid nonce for a sub-problem, it submits that nonce to the controller.
If a nonce leads to a hash at difficulty $D$, the controller forwards this block to the miners outside the pool and
distributes the rewards proportional to the miners who discovered any valid nonce after deducting a management fee.

Mining pools ensure that each miner receives a regular payout (incentive) even if that individual miner cannot discover the
nonce that reaches difficulty $D$.
However, by design, these mining pools sacrifice decentralization for centralized management.
The pool controller receives a fee for managing the pool and
decides the rewards and punishments for the miners, which miners can join the pool, and who to remove from the pool.
Moreover, the existence of mining pools does not eliminate the nature of \PoW{}, 
as often there is more than one mining pool, and these pools compete with each other, which wastes computational resources.

Alternatively, decentralized pools eliminate the need for a pool controller.
However, attempts to create a decentralized mining pool have been unsuccessful due to nodes not doing designated tasks and lack of accountability:
the last block added by any decentralized mining pool in Bitcoin was in 2019~\cite{poolparty,smartpool}.
}
%
%

\section{Power-of-Collaboration}
\label{s:dual}
\PoC{} aims to guard a decentralized system running \PoS/\BFT{} protocols from long-range attacks. 
It offers the following properties:

\begin{enumerate}
    \renewcommand{\labelenumi}{G\arabic{enumi}.}
    \item \label{g:ledger} 
	\PoC{} makes it computationally expensive (solve complex puzzles) for an adversary to overwrite the original ledger,

    \item \label{g:cost} 
	\PoC{} requires miners to collaborate and work on the same block; each miner has to work on a subset of search space.
	Consequently, miners spend less resources than \PoW{}.

    \item \label{g:fair} 
	\PoC{} ensures fairness by distributing incentives among all the miners; 
	even if there is no valid nonce in a miner's search space, 
	it receives an incentive for its efforts.

    \item \label{g:accountability} 
	\PoC{} penalizes any miner that fails to find a valid nonce, if present, in its search space.
\end{enumerate}

Before we describe the design of \PoC{}, we discuss 
some of the possible solutions and their limitations.

{\em Version 1.}
Say a decentralized system running a \PoS/\BFT{} consensus protocol employs a \PoW{} consensus subsystem 
to guard itself against long-range attacks.
Each batch of transactions that the consensus protocol commits is forwarded to the \PoW{} subsystem to add to 
the ledger maintained by \PoW{} miners.
Each miner $\Miner{}$ follows the \PoW{} protocol: 
creates a \PoW{} block that includes one or more committed batches, a transaction that transfers incentive to its account, 
the hash of the previous block $\prev$, and initiates the search for a valid nonce. 
When $\Miner{}$ finds the nonce, it broadcasts its block and the nonce to the other miners.
If another miner $\Miner{'}$ receives a block/nonce, it starts building/mining the next block and includes the hash of the received block/nonce as the previous block.
This solution faces the following two limitations:
(i) all but one miner waste their computational resources (lack of fairness), and
(ii) more than one miner can find a valid nonce, which can temporarily fork the ledger and 
lead to a subsequent increase in the hardness/difficulty of finding the nonce (\S\ref{ss:pow}).


{\em Version 2.}
Next, we replace the \PoW{} consensus subsystem with a centralized mining pool, where the pool operator receives the next 
committed block and creates sub-problems for the pool's miners to mine.
This solution does not face any of the above limitations. 
However, this solution illustrates control by a single operator/organization, which receives fees for its services and decides the incentives/penalties for the pool’s miners.

{\em Version 3.}
Finally, we replace the centralized mining pool with a decentralized mining pool, where miners decide to coordinate with each other 
without any operator. 
The following are the challenges for any decentralized mining pool-based solution
(i) which miner decides the content of the block, 
(ii) how to fairly distribute rewards among the miners, and
(iii) how to detect and penalize a malicious miner that  delays block mining by not performing designated tasks

\begin{figure}[t]
    \centering
    \includegraphics[width=0.4\textwidth]{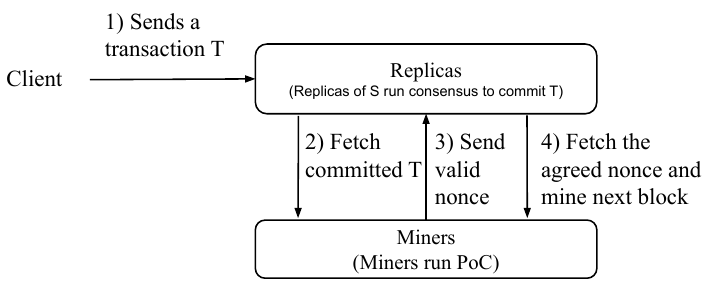}
    \caption{Transactional flow in a system $\Service$+\PoC.}
    \label{fig:poc_overview}
\end{figure}

{\bf Overview.}
Our solution should offer the four appealing properties (G\ref{g:ledger} to G\ref{g:accountability}).  
Consequently, we design \PoC{} that requires no centralized organization and guards a \PoS/\BFT{} protocol 
from long-range attacks.
In Figure~\ref{fig:poc_overview}, we illustrate the transactional flow.

(1) \PoC{} expects that the underlying \PoS/\BFT{} protocol reaches consensus on client transactions among its replicas 
and forwards every committed batch of transactions to the \PoC{} miners.

(2) Once miners are ready to mine, they select a set of ordered batches to form a block. 
To allow miners to collaborate and reduce computational resource wastage, \PoC{} ensures that all the miners are mining identical blocks, and
each miner searches for the valid nonce on a unique subset of the search space (property~G\ref{g:cost}).
\Modify{It also shows that an adversary must provide more resouce than half of the miners in order to create an alternative block 
as all the miners solve the nonce collaboratively. (property~G\ref{g:ledger})}

(3) Once a miner discovers a nonce, it broadcasts the nonce to everyone. 
To ensure that there are no forks of the ledger, \PoC{} requires the underlying \PoS/\BFT{} consensus protocol 
to attest a nonce. 
Note: this recursive dependency helps to {\em quickly select} a valid nonce; 
the same task {\em can be done} by running a consensus on the nonce among the miners.

(4) Once a miner finds a valid nonce, each miner receives an incentive,
which ensures fair reward distribution (property~G\ref{g:fair});

(5) Byzantine miners may not search for nonce in their subset of the search space. 
If the valid nonce is present in this subset, then no miner will ever discover it. 
\PoC{} allows miners to independently discover such malicious attacks and switches honest miners to different subsets to facilitate the discovery of the nonce.

(6) \PoC{} guarantees accountability by penalizing malicious miners that failed to find a valid nonce (property~G\ref{g:accountability}).


Next, we discuss our \PoC{} assuming no attacks ({\em good case}).
Later, we explain how we handle malicious attacks.

\subsection{Client Transaction Ordering}
Prior to running \PoC{}, we expect the blockchain system $\Service$ to reach
consensus on client transactions. 
For \PoC{}, the consensus run by $\Service$ is a black box. 
The system $\Service$ is free to run any consensus protocol of
its choice. 
It needs to only provide \PoC{} miners with {\em committed} \Modify{blocks, a batch of transactions.}

Our use of the term committed implies that each committed \Modify{block}
has been accepted by a quorum of replicas of $\Service$ and will persist across
adversarial failures. For instance, in \BFT{} protocols like \pbft{}~\cite{pbftj} and
\hotstuff~\cite{hotstuff}, blocks are committed when they have quorum
certificates from $2\f{\Replicas}+1$ replicas. In several other blockchain systems,
\Modify{a block} is assumed committed if it is at a specific depth in the blockchain.
A depth indicates the position of the block in the blockchain; the larger the depth of a
block $\block$, the greater the number of blocks that succeed $\block$ and the harder
it is for another fork to overtake this chain.

The assumption of \PoC{} miners working with only committed batches has two advantages:
(1) It frees \PoC{} from having any knowledge on the consensus run by $\Service$, and
(2) \PoC{} miners will not waste their resources on batches that may not persist.

\subsection{Chain Communication}
\label{ss:ida}
Once the blockchain system $\Service$ has committed a block, we require it to forward the
committed block to \PoC{} miners to log it in the ledger.
Like popular blockchain systems, Bitcoin and Ethereum, we employ the {\em gossip} protocol
for broadcasting a block.
Here, we are making a simplifying assumption that each node is connected to a sufficient number of honest nodes because,
in gossip protocols, each node forwards the message to only its neighbors.
Alternatively, the replicas of $\Service$ can employ either the Information Dispersal Algorithm~\cite{rabinida}
or Byzantine Reliable Broadcast~\cite{good-byzantine-broadcast} if they cannot assume a uniform distribution of honest nodes.
Each miner accepts a committed block once it receives it from $\f{\Replicas{}}+1$ replicas in $\Service$,
which assures this miner that it did receive a committed block.

\subsection{Collaborative Mining}
\label{ss:collaboration}
\PoC{} introduces the notion of collaborative mining to log each committed block in the ledger.

Collaborative mining, like centralized pooled mining, should ensure that each miner works on a unique sub-problem so that miners 
do not waste their computational resources. 
Thus, we divide the \PoW{} hash computation into $\n{\Miners}$ disjoint
sub-problems and require each miner to work on a {\em distinct predetermined sub-problem}.
Like existing \PoW{} systems~\cite{bitcoin}, 
\PoC{} miners have to compute a SHA-256 hash, which is represented as a $32$-byte hexadecimal value.
Thus, the solution space $\Slice{}$ comprises of $2^{256}$ possible values.
We divide $\Slice{}$ into $\TotalSlices$ {\em slices}; the size of each slice is a system parameter.
%
Given $\Slice{1}, \Slice{2}, ... , \Slice{\TotalSlices}$ slices, the following holds:
\begin{center}
    $\Slice{1} \cap \Slice{2} \cap \cdots \cap \Slice{\TotalSlices} = \varnothing$\quad and \quad $\Slice{1} \cup \Slice{2} \cup \cdots \cup \Slice{\TotalSlices} = \Slice{}$
\end{center}
\PoC{} assigns each miner one or more consecutive slices based on its stakes.
We assume that each slice is assigned to a miner in $\Miners$.
We discuss the slice assignment scheme in more detail in \S~\ref{ss:staking}.
Given the difficulty $D$, each miner computes a hash till it reaches the target $32$-bit SHA-256 hash
(refer to Equation~\ref{eq:difficulty}).
This requires each miner to find a nonce $\Nonce$ in its set of slices.

Next, we describe the \PoC{} consensus.


\subsection{\PoC{} Protocol Steps}
\label{ss:poc-protocol}
From an outside view, our \PoC{} protocol works in rounds, and within each round, each miner
attempts to find a valid nonce in its pre-determined slices.
Next, we explain the \PoC{} protocol under the assumption that each miner knows
the next block to mine (\S\ref{ss:collaboration}) and has received this block through chain communication (\S~\ref{ss:ida}).
In the case a miner does not have access to the next block to mine,
it can ask the other miners about the missing blocks.

\begin{figure}[t]
    \begin{tikzpicture}[xscale=1.4,list/.style={minimum width=1.8cm,rectangle split, rectangle split parts=2,draw, rectangle split}]

        \node[list] (A) at (-1.2, 0) {\scriptsize \MName{hash}\nodepart{two}\scriptsize \MName{[$\Transaction_1,\Transaction_{100}$]}};
        \node[list] (B) at (0.3, 0) {\scriptsize \MName{hash}\nodepart{two}\scriptsize \MName{[$\Transaction_{101}, \Transaction_{200}$]}};
        \node[list] (C) at (1.8, 0) {\scriptsize \MName{hash}\nodepart{two}\scriptsize \MName{[$\Transaction_{201}, \Transaction_{300}$]}};
        \node[list] (D) at (3.3, 0) {\scriptsize \MName{hash}\nodepart{two}\scriptsize \MName{[$\Transaction_{301}, \Transaction_{400}$]}};
        \node[above] at (A.north) {\scriptsize \MName{Block $1$}};
        \node[above] at (B.north) {\scriptsize \MName{Block $2$}};
        \node[above] at (C.north) {\scriptsize \MName{Block $3$}};
        \node[above] at (D.north) {\scriptsize \MName{Block $4$}};
        \path (D.text west) edge[->] (C.text east)
        (C.text west) edge[->] (B.text east)
        (B.text west) edge[->] (A.text east)
        (A.text west);

        \draw[decoration={brace,amplitude=8pt,mirror},decorate,thick,blue!50!black!90] (-1.1, -0.6) -- (0.2,-0.6);

        \draw[decoration={brace,amplitude=8pt,mirror},decorate,thick,green!50!black!90] (2, -0.6) -- (3.3,-0.6);


        \node[list] (E) at (-0.5, -1.9) {\scriptsize \MName{hash}\nodepart{two}\scriptsize \MName{[$\Transaction_{1}, \Transaction_{200}$]}};
        \node[list] (F) at (2.5, -1.9) {\scriptsize \MName{hash}\nodepart{two}\scriptsize \MName{[$\Transaction_{201}, \Transaction_{400}$]}};
        \node[above] at (E.north) {\scriptsize \MName{Blocks [$1$-$2$]}};
        \node[above] at (F.north) {\scriptsize \MName{Blocks [$3$-$4$]}};
        \path (F.text west) edge[->] (E.text east);

        \node[left] at (1.65, -0.8) {{\small $\Service$-blocks}};
        \node[left] at (1.6, -1.9) {{\small Mined blocks}};

    \end{tikzpicture}
    \caption{Illustrating how  $\Bsize = 2$ contiguous $\Service$-blocks produced by replicas of $\Service$
        are aggregated into one mined block of \PoC{}.
        Here, each $\Service$-block includes $100$ transactions.
    }
    \label{fig:block-batch}
\end{figure}

{\bf Block Creation.}
When a \PoC{} miner receives a block from $\f{\Replicas{}}+1$ replicas, it adds that block to the {\em list of pending blocks}.
The list of pending blocks is an ordered list of committed blocks that a miner is yet to add to the ledger; 
these blocks are ordered by the sequence number assigned by the \PoS/\BFT{} consensus protocol running at $\Service$.
As the difficulty $D$ of mining a block sets the rate at which blocks are added to the \PoC{} ledger,
which in turn impacts the system throughput and latency, at higher difficulties, \PoC{} has a lower throughput than the \PoS/\BFT{} consensus protocol. 
We present a discussion on \PoC's difficulty in~\S\ref{s:discussion}.
Consequently, \PoC{} miners have an ever-growing list of pending blocks, as the rate at which they add these blocks is slower than 
the rate at which they receive them from the \PoS/\BFT{} protocol.

To reduce this gap between total blocks received and blocks mined,
\PoC{} allows miners to batch a set of committed $\Service$-blocks, thereby,
each {\em mined block} includes $\Bsize > 0$ committed
$\Service$-blocks.\footnote{For disambiguation, we use $\Service$-blocks to denote
    the blocks produced by replicas of $\Service$ and mined block to denote the block produced by miners.}
The value of $\Bsize{}$ is a system parameter.
\PoC{} requires miners to only aggregate $\Bsize$ consecutive $\Service$-blocks from the pending list, 
which is essential for maintaining the block ordering by $\Service$. 
For the sake of discussion, let us assume that each $\Service$-block
is assigned a monotonically increasing sequence number $\Sequence$ and each mined block is
assigned a sequence number $\BlockHeight$. If the last block mined by miners had sequence
number $\BlockHeight - 1$ and the sequence number of last $\Service$-block added to the $(\BlockHeight-1)$-th block
is $\Sequence-1$, then in the $\BlockHeight$-th block, each miner will aggregate the following
blocks: $\Sequence, \Sequence+1,..., \Sequence+\Bsize$. 
Aggregating $\Service$-blocks in this way is safe as these blocks are already
committed by replicas of $\Service$.
We illustrate this next.

\begin{example}
    In Figure~\ref{fig:block-batch},
    the replicas of $\Service$ committed four $\Service$-blocks starting with sequence number $1$.
    Assume $\Bsize = 2$, then each \PoC{} miner aggregates $2$ consecutive $\Service$-blocks into a mined block.
\end{example}

Each mined block includes a Merkle root of all the transactions;
as each $\Service$-block contains a Merkle root of all the transactions, 
a miners $\Miner{}$ generates the Merkle root for the mined block by hashing the Merkle roots of all the aggregated blocks.

{\bf Nonce Discovery.}
Once a miner knows the valid nonce for $(\BlockHeight-1)$-th block, it initiates the search for the nonce for $\BlockHeight$-th block.
Assuming the miner $\Miner_{i}$ knows the set of slices it needs to mine (\S\ref{ss:staking}), which we denote as $\Slice{i}$,
$\Miner_{i}$ initiates nonce discovery.
Specifically, $\Miner_{i}$ iterates over all the values in its slices $\Slice{i}$.

    {\bf Nonce Announcement.}
\label{ss:reliable_submition}
When a miner $\Miner_{i}$ finds a valid nonce $\Nonce$, it creates a message \MName{NonceFind}
that includes $\Nonce$ and broadcasts this message to all the miners. 
When a miner $\Miner_{j}$ receives a \MName{NonceFind} message, it terminates
the process of nonce discovery if the received nonce is valid. Next, each miner that has access
to a valid nonce gossips this nonce to the replicas of $\Service$. 

    {\bf Nonce Attestation.}
Next, \PoC{} leverages the underlying \PoS/\BFT{} consensus protocol
to attest the discovered nonce. Specifically, when a replica $\Replica{}$ of $\Service$ receives matching
$\MName{NonceFind}$ messages from $\f{\Miners}+1$ miners, it creates a transaction that includes
the received \MName{NonceFind} message as its data. Whenever it is the turn of $\Replica{}$ to
propose a new block, and if an $\Service$-block containing the $\MName{NonceFind}$ message for the
$\BlockHeight$-th mined block is yet to be proposed, $\Replica{}$ proposes a new block that
includes this message. We expect honest replicas to prioritize nonce transactions over others
to prevent delays in adding newly mined blocks to the ledger.

    {\bf Chain Append.}
\label{ss:chain_append}
When a miner $\Miner_{i}$ receives the valid nonce for the $\BlockHeight$-th block from replicas
of $\Service$, $\Miner_{i}$ appends this block to its local ledger and marks the nonce
discovery process for the $\BlockHeight$-th block as complete. 
Post this, $\Miner_i$ executes the {\em reward} transactions in the block to distribute the reward (\S\ref{ss:rewards}).
Finally, $\Miner_i$ begins mining the next block.
We use the following example to illustrate \PoC{} mining.
\begin{example}
    In Figure~\ref{fig:poc}, the solution space $\Slice{} = [0,5]$ is divided among three miners.
    The three slices are: $\Slice{1} = [0,1]$, $\Slice{2} = [2,3]$, and $\Slice{3} = [4,5]$.
    Assume the valid nonce is $2$ and it lies in the slice of miner $\Miner_2$.
    Once $\Miner_2$ discovers the nonce in its slice, it broadcasts the nonce to all the miners,
    following which each miner requests the replicas of $\Service$ to attest this nonce.
\end{example}

\begin{figure}
    \centering
    \begin{tikzpicture}[yscale=0.3,xscale=0.75]
        \draw[thick,draw=black!75]
        (1.75,   1) edge ++(10, 0)
        (1.75,   2) edge [blue!50] ++(10, 0)
        (1.75,   3) edge ++(10, 0)
        (1.75,   4) edge [yellow!50!black!30]++(10, 0);

        \node[label,below,yshift=3pt] at (2.7, 1) {\scriptsize {Block}};
        \node[label,below,yshift=-4pt] at (2.7, 1) {\scriptsize {Creation}};
        \node[label,below,yshift=3pt] at (4.3, 1) {\scriptsize {Nonce}};
        \node[label,below,yshift=-4pt] at (4.3, 1) {\scriptsize {Discovery}};
        \node[label,below,yshift=3pt] at (5.9, 1) {\scriptsize {Nonce}};
        \node[label,below,yshift=-4pt] at (5.9, 1) {\scriptsize {Announce}};
        \node[label,below,yshift=3pt] at (7.3, 1) {\scriptsize {Nonce}};
        \node[label,below,yshift=-4pt] at (7.3, 1) {\scriptsize {Attest}};
        \node[label,below,yshift=3pt] at (8.8, 1) {\scriptsize {Consensus}};
        \node[label,below,yshift=-4pt] at (8.8, 1) {\scriptsize {in $\Service$}};
        \node[label,below,yshift=3pt] at (10.3, 1) {\scriptsize {Nonce}};
        \node[label,below,yshift=-4pt] at (10.3, 1) {\scriptsize {Attested}};

        \draw[thin,draw=black!75]   (2, 1) edge ++(0, 3)
        (3.5, 1) edge ++(0, 3)
        (5, 1) edge ++(0, 3)
        (6.5, 1) edge ++(0, 3)
        (11, 1) edge ++(0, 3);
        \draw[thin,draw=black!75,dashed](8, 1) edge ++(0, 3);
        \draw[thin,draw=black!75,dashed](9.5, 1) edge ++(0, 3);

        \node at (4.3, 2.5) {
            \renewcommand{\arraystretch}{0.7}
            \begin{tabular}{c}
                {\scriptsize[0, 1]} \\
                {\scriptsize[2, 3]} \\
                {\scriptsize[4, 5]} \\
            \end{tabular}
        };

        \node[left] at (1.8, 1) {\scriptsize $\Miner_3$};
        \node[left] at (1.8, 2) {\scriptsize $\Miner_2$};
        \node[left] at (1.8, 3) {\scriptsize $\Miner_1$};
        \node[left] at (1.8, 4) {\scriptsize $\PBFT{}$};

        \path[->] (2, 4)
        edge (3.5, 3)
        edge (3.5, 2)
        edge (3.5, 1);
        \path[->] (5, 2)
        edge (6.5, 3)
        edge (6.5, 1);
        \path[->] (6.5, 3)
        edge (8, 4);
        \path[->] (6.5, 2)
        edge (8, 4);
        \path[->] (6.5, 1)
        edge (8, 4);
        \path[->] (9.5, 4)
        edge (11, 3)
        edge (11, 2)
        edge (11, 1);

    \end{tikzpicture}
    \caption{\PoC{} protocol with miners $\Miners{} = \{\Miner_1, \Miner_2, \Miner_3\}$.
    The solution space $\Slice{} = [0,5]$ is divided into three
    slices $([0,1], [2,3], [4,5])$. 
    Assume the valid nonce $2$ and miner $\Miner_2$ discovers it in its slice.
    }
    \label{fig:poc}
\end{figure}

    {\bf Multiple Nonces.}
In rare scenarios, two or more miners may find valid nonces that
help to reach the expected hash. Specifically, miners $\Miner_{i}$ and $\Miner_{j}$ may both
discover a valid nonce in their respective slices. This situation is {\em not unique} to \PoC{},
even complex \PoW{} computations can have multiple solutions.
We need to guarantee that all the miners select the same nonce, which
is trivial for \PoC{} as each nonce is attested by replicas of $\Service$.
If the next
proposer for an $\Service$-block receives two or more $\MName{NonceFind}$ messages with distinct
valid nonces, it selects one of them as the solution. When, eventually, this block becomes committed,
all the \PoC{} miners will have access to the same nonce for block $\BlockHeight$.

\subsection{Staking and Rewards}
\label{ss:rewards}
\label{ss:staking}
\PoC{}, like \PoW{}, rewards its miners with incentives for their participation in the mining process; 
miners expend their computational resources and would only do so if they make some profit.
However, unlike \PoW{}, \PoC{} wants to ensure fairness by rewarding each miner for collaboration even though the valid nonce was not in its slice. 
In \PoC{}, we reward each miner in proportion to the number of slices in its slice set.
As stated in Section~\ref{ss:collaboration}, in \PoC{}, there are a total of $\TotalSlices$ slices.
We assume the following: (1) Each of these $\TotalSlices$ slices is assigned to a unique miner.
(2) Each miner is assigned a set of consecutive slices.
We require a miner to invest its monetary resources in exchange for each slice it holds.
Like existing \PoS{} systems, we term this investment by a miner as {\em staking} as the miner
no longer has access to its invested monetary resources~\cite{ether,algorand,cardano}.
\Modify{If a miner could not finish the mining due to the excessive slice space, it will be  regarded as an adversary.}
\Modify{We will also disscuss the miner reconfiguration in ~\S\ref{s:discussion} to reassign the slices Whenever a new miner joins or an old miner leaves the system.}

The economics of converting actual currency into an online tradable commodity is a problem faced by every decentralized system, for which current literature includes several ad hoc solutions. 
Thus, its design is beyond the scope of this paper.
For simplicity, we assume that \PoC{} builds on top of some token $\Token$, where $\Token$ is the cost of purchasing a slice.
Each miner exchanges its currency for a set of $\Token$.
If a miner $\Miner_i$ wants $e$ slices in its set, $\Miner_i$ stakes $e \times \Token$ tokens.
This information is added to the first block of the \PoC{} ledger, which is often referred to as the genesis block.
Specifically, \PoC-genesis blocks stores the following information:
(1) total number of miners ($\n{\Miners{}}$),
(2) number of tokens staked by each miner,
(3) a public key for each miner, and
(4) slice to miner mapping.

During \PoC's collaborative mining, each miner refers to this genesis block to identify
the slices it is responsible for mining.
Once a miner receives a valid nonce from the replicas of $\Service$, it adds a reward to the
account of each miner.
Like existing systems~\cite{ether,bitcoin}, we assume that the reward is proportional to the fees paid by the clients;
each client pays a fee to add its transaction to the ledger and this fee is divided among the miners in proportion to their number of slices.
For example, if the client for transaction $\Transaction$ pays $\Diamond$ tokens and
a miner $\Miner_i$ has $e$ slices in its set, $\Miner_i$ receives $\frac{e \times \Diamond \times \Token}{\TotalSlices}$ tokens as a reward.
We maintain each miner's account as a key-value pair in a NoSQL database replicated across all miners/replicas;
the key field represents the public key (logged in the genesis block) of each miner, while the value field represents the token balance.

\section{Malicious Attacks}
\label{s:attacks}
Unlike a centralized mining pool, where the pool operator oversees the activity of all the miners and rewards/penalizes them for their actions, 
\PoC{} assumes a decentralized setup. 
Thus, \PoC{} needs to provide protection from malicious attacks and guarantee accountability
while ensuring that it makes it computationally expensive for the adversary to overwrite the ledger using long-range attacks.

    {\bf First}, we expect that the underlying \PoS/\BFT{} protocol guarantees safety and liveness properties~\S\ref{s:discussion}. 
    Thus, it should be capable of handling attacks by a standard adversary.

    {\bf Second}, a standard adversary can only attack the \PoC{} in the following ways:
\begin{enumerate}[nosep]
    \renewcommand{\labelenumi}{A\arabic{enumi}.}
    \item  \label{a4}As multiple nonces can satisfy Equation~\ref{eq:difficulty} and
          if Byzantine miners are fortunate enough to find two such nonces for a mined block,
          they can equivocate by sending each nonce to only a subset of honest miners and replicas.

    \item  \label{a1}A miner decides to not participate in the \PoC's collaborative mining or if a miner
          discovers a nonce in its slice, it decides to not send it to other miners.
\end{enumerate}

{\bf Third}, an advanced adversary that has access to the keys of any replica/miner can 
use these keys to forge any message.

    {\bf Final}, each distributed system
needs to also guard against denial-of-service attacks. To prevent these attacks,
we follow the best practices suggested by prior works~\cite{aadvark,upright} and
assume that the replicas/miners use one-to-one virtual communication channels, which
can be disconnected if needed.

Next, we discuss how we handle attacks by a standard adversary on \PoC{} and
attacks by advanced adversaries on the entire system.

\begin{figure*}
    \centering
    \begin{tikzpicture}[yscale=0.3,xscale=0.75]
        \draw[thick,draw=black!75] (0.75,   0) edge ++(21.5, 0)
        (0.75,   1) edge ++(21.5, 0)
        (0.75,   2) edge [red] ++(21.5, 0)
        (0.75,   3) edge ++(21.5, 0)
        (0.75,   4) edge [yellow!50!black!30]++(21.5, 0);

        \node[label,below,yshift=3pt] at (1.5, 0) {\scriptsize {Block [$\block_{1}$]}};
        \node[label,below,yshift=3pt] at (3.2, 0) {\scriptsize {Nonce}};
        \node[label,below,yshift=-4pt] at (3.2, 0) {\scriptsize {Discovery}};
        \node[label,below,yshift=3pt] at (4.7, 0) {\scriptsize {Slice}};
        \node[label,below,yshift=-4pt] at (4.7, 0) {\scriptsize {Shift}};
        \node[label,below,yshift=3pt] at (6.25, 0) {\scriptsize {Shift}};
        \node[label,below,yshift=-4pt] at (6.25, 0) {\scriptsize {Certificate}};
        \node[label,below,yshift=3pt] at (7.74, 0) {\scriptsize {Consensus}};
        \node[label,below,yshift=-4pt] at (7.74, 0) {\scriptsize {in $\Service$}};
        \node[label,below,yshift=3pt] at (9.2, 0) {\scriptsize \MName{Shift}};
        \node[label,below,yshift=3pt] at (10.65, 0) {\scriptsize {Nonce}};
        \node[label,below,yshift=-4pt] at (10.65, 0) {\scriptsize {Discovery}};
        \node[label,below,yshift=3pt] at (12.3, 0) {\scriptsize {Nonce}};
        \node[label,below,yshift=-4pt] at (12.3, 0) {\scriptsize {Announce}};
        \node[label,below,yshift=-4pt] at (13.8, 0) {\scriptsize {Nonce}};
        \node[label,below,yshift=3pt] at (13.8, 0) {\scriptsize {Attest}};
        \node[label,below,yshift=3pt] at (15.2, 0) {\scriptsize {Consensus}};
        \node[label,below,yshift=-4pt] at (15.2, 0) {\scriptsize {in $\Service$}};
        \node[label,below,yshift=3pt] at (16.7, 0) {\scriptsize {Nonce}};
        \node[label,below,yshift=-4pt] at (16.7, 0) {\scriptsize {Attested}};
        \node[label,below,yshift=3pt] at (18.35, 0) {\scriptsize {Penalty}};
        \node[label,below,yshift=3pt] at (19.9, 0) {\scriptsize {Penalty}};
        \node[label,below,yshift=-4pt] at (19.9, 0) {\scriptsize {Certificate}};
        \node[label,below,yshift=3pt] at (21.3, 0) {\scriptsize {Consensus}};
        \node[label,below,yshift=-4pt] at (21.3, 0) {\scriptsize {in $\Service$}};

        \draw[thin,draw=black!75]   (1, 0) edge ++(0, 4)
        (2.5, 0) edge ++(0, 4)
        (4, 0) edge ++(0, 4)
        (5.5, 0) edge ++(0, 4)
        (10, 0) edge ++(0, 4)
        (11.5, 0) edge ++(0, 4)
        (13, 0) edge ++(0, 4)
        (17.5, 0) edge ++(0, 4)
        (19, 0) edge ++(0, 4);

        \draw[thin,draw=black!75,dashed]
        (7, 0) edge ++(0, 4)
        (8.5, 0) edge ++(0, 4)
        (14.5, 0) edge ++(0, 4)
        (16, 0) edge ++(0, 4)
        (20.5, 0) edge ++(0, 4)
        (22, 0) edge ++(0, 4);

        \node[left] at (0.8, 1) {\scriptsize $\Miner_3$};
        \node[left,red] at (0.8, 2) {\scriptsize $\Miner_2$};
        \node[left] at (0.8, 3) {\scriptsize $\Miner_1$};
        \node[left] at (0.8, 4) {\scriptsize $\PBFT{}$};

        \node at (3.2, 2.5) {
            \renewcommand{\arraystretch}{0.7}
            \begin{tabular}{c}
                {\scriptsize \MName[0, 1]} \\
                {\scriptsize \MName[2, 3]} \\
                {\scriptsize \MName[4, 5]} \\
            \end{tabular}
        };

        \path[->] (1, 4)
        edge (2.5, 3)
        edge (2.5, 2)
        edge (2.5, 1);
        \path[->] (4, 3)
        edge (5.5, 3)
        edge (5.5, 2)
        edge (5.5, 1);
        \path[->] (4, 1)
        edge (5.5, 3)
        edge (5.5, 2)
        edge (5.5, 1);
        \path[->] (5.5, 3)
        edge (7, 4);
        \path[->] (5.5, 1)
        edge (7, 4);
        \path[->] (8.5, 4)
        edge (10, 3)
        edge (10, 2)
        edge (10, 1);
        \path[->] (11.5, 3)
        edge (13, 2)
        edge (13, 1);
        \path[->] (13, 3)
        edge (14.5, 4);
        \path[->] (13, 1)
        edge (14.5, 4);
        \path[->] (16, 4)
        edge (17.5, 3)
        edge (17.5, 2)
        edge (17.5, 1);
        \path[->] (17.5, 3)
        edge (19, 2)
        edge (19, 1);
        \path[->] (17.5, 1)
        edge (19, 3)
        edge (19, 2);
        \path[->] (19, 3)
        edge (20.5, 4);
        \path[->] (19, 1)
        edge (20.5, 4);
        \node at (10.7, 2.5) {
            \renewcommand{\arraystretch}{0.7}
            \begin{tabular}{c}
                {\scriptsize \MName[2, 3]} \\
                {\scriptsize \MName[4, 5]} \\
                {\scriptsize \MName[0, 1]} \\
            \end{tabular}
        };
    \end{tikzpicture}
    \caption{Slice shifting procedure: assume $2$ is the valid nonce and is present in the slice
        of malicious miner $\Miner_2$. $\Miner_2$ fails to broadcast $2$, which triggers slice shifting;
        and with help of replicas of $\Service$, once $2$ is found, it is penalized.
    }
    \label{fig:poc_malicious}
\end{figure*}

\subsection{Standard Adversary}
As replicas of $\Service$ help in selecting a nonce, Attack~A\ref{a4} is trivially resolved.
Each replica $\Replica{}$ selects a nonce for the mined block at height $\BlockHeight$
only after it receives $\f{\Miners}+1$ matching \MName{NonceFind} messages.
Whenever it is the turn of $\Replica{}$ to propose a new block,
and if an $\Service$-block containing the nonce for the $\BlockHeight$-th mined block is yet to be proposed,
$\Replica{}$ proposes a new block that includes this nonce.
Thus, only one of the two nonces will commit and honest miners will receive only one nonce.
If there aren't sufficient matching messages, then no nonce will be selected and
this will lead to the eventual detection of some Byzantine miners.

To resolve Attack~A\ref{a1}, next, we present our {\em slice shifting} protocol,
which penalizes miners for malicious behavior.

\subsection{Slice Shifting}
\label{ss:slice-shifting}
Each honest \PoC{} miner should search for a nonce in its set of slices until it receives a
valid nonce from another miner or it has exhausted its slices. A malicious miner may not follow this
behavior; it may want to disrupt the process of nonce finding. If the malicious miners are fortunate
and the nonce is present in their slices, then honest miners may never receive the nonce, and \PoC{} will come to a halt. 
We aim to quickly resolve such a situation. 
We do so by running our {\em slice shifting} algorithm that switches
miner slices under failures. This is done with the aim that when an honest miner has access to the
slice held by a malicious miner, it can discover the nonce in that slice and broadcast it to other miners.

Next, we illustrate slice shifting protocol through an example.
Post this, we will explain the protocol in detail.

\begin{example}
    Figure~\ref{fig:poc_malicious} illustrates the slice shifting protocol. Assume the nonce lies in
    slice $[2,3]$ and miner $\Miner_2$ fails to announce the nonce. Eventually, honest miners
    timeout while waiting to receive a valid nonce and trigger slice shifting protocol. These miners must create a certificate to prove that
    a majority wants to do slice shifting. They send this certificate to the replicas of $\Service$
    for attestation. Post this, each miner works on a new slice. Once, they discover the nonce, they
    initiate the process of penalizing $\Miner_2$.
\end{example}

{\bf Timer Initialization.}
\PoC{} miners, who have exhausted their slice search space and do not have a valid nonce,
need a mechanism to make progress. We follow existing \BFT{} works and require each miner to set
a {\em timer} before it starts mining a slice. Specifically, each miner sets a timer $\delta$ for
the $\BlockHeight$-th block 
and stops $\delta$ when it has a valid nonce for the $\BlockHeight$-th block.

    {\bf Malicious Miner.}
If $\Miner{}$'s timer $\delta$ expires and it does not have access to a valid nonce, $\Miner{}$
announces to all the other miners that it wishes to initiate the {\em slice shifting} protocol. 
The slice shifting protocol runs for at
most $\f{\Miners{}}$ rounds and deterministically switches slices assigned to each miner.
A round of slice shifting only takes place when at least $\f{\Miners{}}+1$ miners request to do so.
Specifically, when a miner $\Miner_{i}$ timeouts, it creates a message
$\Message{\MName{Shift}}{\BlockHeight, \ShiftRound}$ and broadcasts this message to all the other
miners. Here, $\ShiftRound$ represents the slice shifting round, which is initially set to
    {\em zero}. When $\Miner{}$ receives $\MName{Shift}$ messages from $\f{\Miners{}}+1$ distinct miners,
it requests the replicas of system $\Service$ to help reach a consensus on slice shifting. 
To make this request, each miner must broadcast
a certificate $\Certificate$ that includes $\f{\Miners{}}+1$ $\MName{Shift}$ messages.

    {\bf Shift Attestation.}
When replicas of $\Service$ receive a signed $\Certificate$ from $\f{\Miners{}}+1$ \PoC{} miners,
they agree to attest this slice shift. This attestation requires the replicas to run consensus on this
certificate $\Certificate$; 
the next proposer for a $\Service$-block includes $\Certificate$ as a transaction in its block.
{\em Note:} consensus on $\Certificate$ is like consensus on any transaction where $\Certificate$ acts
as the transactional data. Post consensus, all the replicas
gossip this block to the miners.
Once a miner $\Miner{}$ receives a $\Service$-block from $\f{\Replicas}+1$ replicas that include
a committed certificate $\Certificate$, it assumes it is time to shift its slices. Following this,
it increments the shift round $\ShiftRound$ by one and mines the next slice. 
If in shift round $\ShiftRound$, $\Miner_i$ was responsible for mining slices
$\{ \Slice{i}, \Slice{i+1}, ... , \Slice{j} \}$, in round $\ShiftRound+1$,
$\Miner_i$ will mine slices $\{ \Slice{i+1}, ... , \Slice{j}, \Slice{o} \}$,
where $o = (j+1) \mod \TotalSlices$, and $\TotalSlices$ is the total number of slices.
Again, before mining for round $\ShiftRound+1$, $\Miner_i$
restarts the timer $\delta$ for the $k$-th block.
It is possible that the timer $\delta$ again timeouts, due to more failures. In such a case, $\Miner_i$ would
need to initiate another round of slice shifting.
However, under a standard adversary and reliable network,
for each mined block, we need to run the slice shifting protocol only $\ShiftRound = \f{\Miners}$ times.

{\bf No nonce -- Merge.}
Although infrequent, \PoC{} miners may encounter cases where no nonce satisfies Equation~\ref{eq:difficulty}. 
This no nonce situation is scarce in \PoW{} because miners work on different blocks, but it is possible in \PoC{} because all the 
miners collaborate on the same block.
We resolve this situation as follows.

If even after $\ShiftRound = \f{\Miners{}}$ rounds honest miners do not have access to a valid nonce,
we require the miners to terminate their search for the nonce and initiate the {\em merge} process, which
requires miners to mine two or more consecutive mined blocks together,
with the hope that mining multiple blocks together increases the probability of finding a nonce.
For instance, for the $\BlockHeight$-th block, if a miner $\Miner{}$ receives a certificate $\Certificate$
from $\Service$ replicas, which has shift round $\ShiftRound = \f{\Miners{}}$, $\Miner{}$ concludes that
no valid nonce exists for the $\BlockHeight$-th block. 
Following this, $\Miner{}$ creates a new block that merges contents of the $\BlockHeight$-th and $(\BlockHeight+1)$-th blocks.
This merged block now
serves as the $\BlockHeight$-th block and miners attempt to find the nonce for this block. The merged block
includes a Merkle root, which is the hash of all the transactions in $\BlockHeight$ and $(\BlockHeight+1)$-th blocks.

    {\bf Penalty for Slice Shifting.}
Frequent slice shifting due to malicious miners is detrimental to the performance of \PoC{};
it forces honest miners to do more work and wastes their resources. To discourage these attacks, we
{\em penalize} malicious miners. We require each miner to track the number of
shifts ($\ShiftRound$) it took to find a valid nonce and to identify the miners that failed to perform this task.

In \PoC{}, identifying malicious miners responsible for $\ShiftRound$ shifts is a trivial
task for honest miners.
When a miner $\Miner_{i}$ receives a valid nonce for the $\BlockHeight$-th block
in shift round $\ShiftRound$ ($1 \leq \ShiftRound \leq \f{\Miners{}}$), it identifies malicious miners
based on the slice containing the valid nonce.
Let $\Slice{j}$ be the slice, then $\Miner_{i}$ looks into the
genesis block to find the initial assignment of the slice.
Using this information, $\Miner_{i}$ determines the miners who held this slice in subsequent $\ShiftRound - 1$ rounds
and adds these miners to the set of malicious miners $\Miners_{mal}$.

To penalize these malicious miners, we again invoke the replicas of $\Service$. Once a
miner $\Miner{}$ has the knowledge of set $\Miners_{mal}$, it sends a
message $\Message{\MName{Penalty}}{\Miners_{mal}, \BlockHeight, \ShiftRound}$ to all the miners.
Once $\Miner{}$ receives \MName{Penalty} messages from $\f{\Miners{}}+1$ miners,
it creates a certificate (like it created for slice shifting) that includes these \MName{Penalty}
messages and broadcasts this certificate to each replica in $\Service$.

When the replicas of $\Service$ receive a $\MName{Penalty}$ certificate signed by
$\f{\Miners{}}+1$ distinct miners, they add it to a new block for consensus.
Post consensus, honest miners penalize malicious miners by deducting their account balances.

\subsection{Advanced Adversary: Long-Range Attacks}
\label{ss:long-range}
An advanced adversary, unlike a standard adversary, has access to the private keys of honest replicas/miners.
It can use these keys to forge any message that has a digital signature.
In the context of our system, the following attacks are possible:
(1) Multiple committed $\Service$-blocks with the same sequence number,
(2) Forgery of nonce, shift, or penalty messages.
(3) Forgery of $\Service$-blocks, containing nonce, shift, or penalty transactions.
If an adversary can compromise the keys of honest replicas/miners, 
\PoC{} {\em can not guarantee} fairness and accountability; honest miners may get penalized.
However, \PoC{} guarantees that it is computationally expensive for the adversary to rewrite the existing ledger.

To illustrate that \PoC{} guards against long-range attacks, we theoretically analyze its hardness.
We follow Example~\ref{ex:long-range} where starting from the $k$-th block,
malicious parties want to rewrite the ledger.
To do so, malicious replicas would first create an alternate set of $\Service$-blocks
using the compromised private keys of honest parties.
Next, these malicious replicas would forward these $\Service$-blocks
to the malicious miners, who will attempt to forge the \PoC{} ledger by creating new mined blocks.
These malicious miners will not receive any help from honest miners as they have a local copy of the ledger.

As the combined computational power of honest miners is more than malicious miners,
it is safe to assume that in any round of \PoC{},
malicious miners control at most $50\%$ slices and have a $50\%$ chance of finding a valid nonce.%
\footnote{
Assuming that an adversary has access to the private keys of honest nodes is orthogonal to the assumption that an adversary also has more computational power than honest nodes. If an adversary controls more computational power, then it can rewrite even the \PoW{} ledger.
}
So, $\frac{1}{2}$ is the probability that malicious miners find the solution for an ``alternate block'' in their set of slices.
The probability that malicious miners find the solution for the $b$ consecutive alternate blocks in their set of slices is $(\frac{1}{2})^b$.
If $b \ge 7$, the probability is $0.7\%$.
Clearly, if Byzantine miners have to create a large number of alternate blocks, it is impossible for them
to find the nonces by just mining their own set of slices.
Alternatively, they can search for the nonce in all the slices (entire solution space)
but this at least doubles their work.

Next, we measure the actual time required to create an alternate
chain by taking into account two parameters: (1) the age of the chain $\alpha$ (in months)
and (2) the hashing power ratio of the malicious miners ($m$) over the honest miners ($h$), $\frac{m}{h}$.
Essentially, it takes $\alpha \times \frac{h}{m}$ months for the malicious miners to reconstruct
an $\alpha$-month-old chain given their overall hashing power ratio of $\frac{m}{h}$. For example,
if the malicious miners hold $\frac{1}{2}$ of the mining power and want to reconstruct a 1-year
old chain, it would require two years' worth of computation to create an alternative chain of equal
length.

During these two years, the following two things can happen:
(1) The original chain will continue growing, which further increases the task of malicious miners.
(2) The system is at a stall and honest miners/replicas detect that nothing is getting appended to the ledger and will
leave the system.
%
These arguments help us to demonstrate that simply compromising the
private keys of honest miners/replicas is insufficient to launch a long-range attack on \PoC{}.

{\bf Proofs.}
For detailed analysis and proof of correctness, please refer to ~\S\ref{sec:proofs}.

\section{Discussion}
\label{s:discussion}

{\bf Difficulty.}
\label{ss:difficulty}
As discussed in \S~\ref{ss:pow}, difficulty refers to the hardness of finding a valid nonce. 
For \PoW{}, its difficulty depends on the following two parameters:
(1) the number of miners and the hardware technology available to the miners, and
(2) the probability of ledger forks.
In \PoC{}, replicas work on the same block, under the standard adversary model, so there is no possibility of forks.
Consequently, we need to increase/decrease \PoC's difficulty based on the number of miners and the characteristics of the latest hardware technology.

{\bf Miner Reconfiguration.}
\label{ss:reconfiguration}
\PoC{} assumes that reconfigurations are possible; we allow old miners to leave and new miners to join the system.
Like most \PoS/\BFT{} systems~\cite{algorand-stake-withdrawl,ethereum-stake-withdrawl,hotstuff,diembft}, 
we assume that 
(1) each miner only leaves at the boundary of consensus; no miner leaves during an ongoing consensus.
(2) despite reconfigurations, less than $50\%$ of the total number of miners are malicious.
(3) reconfigurations take place under a reliable network.

If a miner $\Miner$ wants to join or leave \PoC{} mining, it creates a message
$\SignMessage{\Message{JoinMiner}{pk, \Token}}{\Miner}$ and $\SignMessage{\Message{LeaveMiner}{pk}}{\Miner}$, respectively, and broadcasts this message to
all the replicas of $\Service$. 
We use $pk$ to denote the public key of the miner, and $\Token$ to denote the monetary resources a new miner wants to stake.
$\Token$ also helps to calculate the total number of slices for the new miner.

The next proposer (or leader) on receiving a join/leave request, needs to do the following: 
(1) create a transaction that includes the join/leave message received from the miner,
(2) check if the difficulty of the system needs to be changed and if so, create a 
transaction that includes the updated difficulty of the system. 
(3) redistributes the slices among the miners and creates transactions that map each miner to its set of slices.
The proposer collects these transactions and proposes a special $\Service$-block.
Once replicas reach a consensus on this special block, they forward it to the miners.
When miners finish appending this $\Service$-block to the ledger, 
they create/delete accounts for the joining/leaving miner and update their set of slices.
Note: the miner who sent a $\MName{LeaveMiner}$ message needs to participate till the end of this step, 
and only then will its stake be released and it can leave the system..

{\bf Challenges for \PoC{}.}
Although \PoC{} helps to meet the four desirable properties~G\ref{g:ledger} to~G\ref{g:accountability}, it faces two challenges:

(1) {\em Fortunate miner.} 
As discussed earlier, a Byzantine miner $\Miner{}$ may not follow the \PoC{} protocol and skip searching for nonce in its set of slices. 
If $\Miner{}$ is fortunate and a valid nonce exists in the slice of an honest miner, 
then $\Miner{}$ will receive a reward without doing any work.
Due to the collaborative nature of \PoC{}, it is impossible to detect such a malicious miner.

(2) {\em Restricted access than \PoW{}.}
As described in ~\S\ref{ss:staking}, despite requiring miners to solve a puzzle like \PoW{}, 
\PoC{} requires miners to stake their resources and undergo a reconfiguration protocol before they can join or leave the \PoC{} mining.

\section{Security Arguments} \label{sec:proofs}
We now state and prove the security properties of \PoC{} under the standard adversary model.
We assume that \PoC{} is appended to a decentralized system $\Service$, which provides the
safety and liveness guarantees stated.
\PoC{} adopts the same partial synchrony model adopted in most consensus systems~\cite{pbftj,hotstuff}.
It guarantees {\em safety} in an asynchronous environment where messages can get lost, delayed or duplicated.
However, {\em liveness} is only guaranteed during the periods of synchrony~\cite{pbftj,sbft,hotstuff}.
\PoC{} guarantees {\em fairness} for rewards and penalties only under synchrony.

\textbf{Safety Argument.}
We argue the safety of \PoC{} by relying on the agreement property of the system $\Service$. Intuitively, \PoC{} ensures that every miner witnesses an identical sequence of discovered nonces. This coherence is achieved by sequencing each discovered nonce into the underlying \PoS/\BFT{} protocol. Consequently, irrespective of the number of solutions found to the Proof of Work (PoW) puzzle, all miners eventually converge to a unified nonce\footnote{We note that Ethereum~\cite{ether} uses a similar technique to ensure that all replicas observe the same set of signatures generated by its finality gadget~\cite{goldfish}.}.

\begin{theorem}[Safety] \label{th:safety}
	No two conflicting blocks are settled by the \PoC{} protocol. That is, if two honest miners $\Miner_i$ and $\Miner_j$
	add blocks $\block_i$ and $\block_j$ at sequence number $\BlockHeight$, then $\block_i = \block_j$.
\end{theorem}
\begin{proof}
	Assume that honest miners $\Miner_i$ and $\Miner_j$ added two conflicting blocks $\block_i$ and $\block_j$ at sequence number $\BlockHeight$.
	We know that each mined block has $\Bsize$ contiguous $\Service$-blocks (\S\ref{ss:poc-protocol}).
	So each mined block, including $\block_i$ and $\block_j$ have the same number of $\Service$-blocks.
	As both conflicting blocks $\block_i$ and $\block_j$ have the same sequence number $\BlockHeight$, it is safe to assume that mined blocks at
	sequence number $ \le \BlockHeight-1$ have the same set of $\Service$-blocks.
	This implies that $\block_i$ and $\block_j$ have at least one distinct $\Service$-block.
	Miners $\Miner_i$ and $\Miner_j$ must have received this distinct $\Service$-block from replicas of $\Service$.
	This can only happen if replicas of $\Service$ committed two distinct blocks at the same sequence number.
	But, this is a contradiction as we assume that $\Service$ guarantees a globally consistent view of the transactions.

	Similarly, each miner submits a $\MName{NonceFind}$ message to replicas in $\Service$ when it has access to a valid nonce.
	The replicas of $\Service$ add this message as a transaction in the next $\Service$-block.
	Assume that this $\MName{NonceFind}$ message is in blocks $\block_i$ and $\block_j$ and it is the transaction
	on which miners differ.
	If this is the case, then replicas of $\Service$ committed different nonce in their $\Service$-blocks,
	which contradicts our assumption on $\Service$.
\end{proof}

\textbf{Liveness Argument.}
We argue the liveness of \PoC{} during periods of synchrony.

\begin{lemma}[Commit Availability] \label{th:commit-availability}
	An honest miner eventually receives the $k$-th $\Service$-block committed by an honest replica.
\end{lemma}
\begin{proof}
	We argue this lemma by induction over the serialized communication of committed $\Service$-blocks.
	Assuming a history of $k-1$ committed $\Service$-blocks for
	which this property holds, we consider the $k$-th committed $\Service$-block.
	When an honest replica has a committed $\Service$-block, it reliably broadcasts that block to the miners.
	It is thus guaranteed that an honest miner will receive all the committed $\Service$-blocks.
	The inductive base case involves assuming that all replicas are initialized with a committed genesis ($k=1$) block, which we can ensure axiomatically.
\end{proof}

\begin{lemma}[Nonce Search] \label{th:nonce-search}
	The first time an honest miner $\Miner_i$ obtains the $\BlockHeight$-th block containing $\sigma > 0$ committed $\Service$-blocks,
	it searches for a valid nonce $\Nonce$ in slice $\Slice{i}$.
\end{lemma}
\begin{proof}
	Upon receiving the $k$-th committed $\Service$-block for the first time,
	correct miners
	%
	wait to check if there are $\sigma > 0$ $\Service$-blocks
	available, since the last $\Service$-block added to the ledger, for mining.
	If so, $\Miner_i$ aggregates these $\Service$-blocks into mined block $b$ and starts to search for a nonce in its set of slices.
\end{proof}


\begin{lemma}[Shift Liveness] \label{th-shift-liveness}
	If a correct miner does not find a valid nonce $\Nonce$ in slice $\Slice{i}$ to settle block $b$ within time $\delta$, another miner eventually tries it.
\end{lemma}
\begin{proof}
	While finding nonce, if the timer $\delta$ expires for an honest miner,
	it broadcast a message $\textsc{Shift}$ to other miners.
	When it receives $\textsc{Shift}$ message from $\f{\Miners}+1$ miners
	it creates a certificate $\Certificate$ comprising of these $\MName{Shift}$
	messages and broadcasts this $\Certificate$ to all the replicas.
	When an honest replica receives $\Certificate$ from at least $\f{\Miners}+1$,
	it proposes it in the next block.
	The liveness property of $\Service$ ensures that these messages are
	eventually committed, and \Cref{th:commit-availability} ensures that
	honest miners are eventually notified of the commit.
	On receiving a committed $\Service$-block with a certificate $\Certificate$,
	an honest miner $\Miner_i$
	resets its timer
	and restarts the nonce finding process in
	slices $\{ \Slice{i+1}, ... , \Slice{j}, \Slice{o} \}$,
	where $o = (j+1) \mod \TotalSlices$ ($\TotalSlices$ is the total number of slices)
	if it has not already attempted to find a nonce for the block for $\f{\Miners}+1$.
	Otherwise, it assumes that no nonce can be found for this block and starts a new mining.
	As a result, honest miners keep shifting and searching for each other nonces until they are all found.
\end{proof}

\begin{lemma}[Block Settlement] \label{th:block-settlement}
	All honest miners settle block $b$ if a valid nonce $\Nonce$ for block $b$ exists.
\end{lemma}
\begin{proof}
	When an honest miner $\Miner_i$ finds a nonce $\Nonce$,
	it broadcasts $\Nonce$ to the other miners. 
	Following this, each miner submits $\Nonce$ to the replicas of $\Service$.
	The liveness property of $\Service$ ensures that $\Nonce$ is eventually committed
	by all honest replicas.
	\Cref{th:commit-availability} then ensures that all honest miners obtain the
	corresponding committed $\Service$-block.
	We conclude the proof by noting that if the $\Service$-block contains a valid
	nonce $\Nonce$, correct miners will settle $b$. 
\end{proof}

\begin{theorem}[Liveness] \label{th:liveness}
	Each committed block at $\Service$ is eventually settled by \PoC{}.
\end{theorem}
\begin{proof}
	Through the liveness property of $\Service$ and \Cref{th:commit-availability}, we conclude that honest miners eventually
	obtain a block committed by replicas of $\Service$.
	\Cref{th:nonce-search} then ensures that miners search for a valid nonce $\Nonce$ to settle this committed $\Service$-block as part of a block $b$.
	Finally, an honest miner can find a nonce in its slice $\Slice{i}$ within time $\delta$ with non-zero probability.
	If it doesn't, \Cref{th-shift-liveness} ensures that honest miners will try again until they succeed.
	As a result, honest miners eventually find a nonce $\Nonce$ for block $b$ in their slice within time $\delta$.
	\Cref{th:block-settlement} then ensures that honest miners use $\Nonce$ to settle block $b$ and thus the committed $\Service$-block.
\end{proof}


\section{Fairness Argument.}\label{s:fair_argument}
We argue that \PoC{} is fair and no correct miners are penalized during periods of synchrony.

\begin{lemma}[Penaly Certificate] \label{th:penalty-certificate}
	There cannot be a penalty certificate $\SignMessage{\Message{Penalty}{\BlockHeight, \ShiftRound{}, \Certificate}}{\Miners}$ unless the timer $\delta$
	of at least one honest miner expires.
\end{lemma}
\begin{proof}
	We start by assuming that there exists a penalty certificate $\SignMessage{\Message{Penalty}{\BlockHeight, \ShiftRound{}, \Certificate}}{\Miners}$
	and 
	the timers of none of the honest miners have expired.
	As each penalty certificate needs signatures of at least $\f{\Miners}+1$ miners and at most $\f{\Miners}$ miners are malicious,
	such an assumption is a contradiction.
\end{proof}

\begin{theorem}[Fairness] \label{th:incentives}
	No honest miner receives a penalty if (i) it can find a nonce $\Nonce$ within time $\delta$ and (ii) the network is experiencing a period of synchrony.
\end{theorem}
\begin{proof}
	Let's assume an honest miner $\Miner_i$ receives a penalty based on shift round $r$. 
	This implies the existence of a penalty certificate including miner $\Miner_i$ in its list of miners to penalize.
	\Cref{th:penalty-certificate} states that this certificate can only exist if the timer $\delta$ of at least {\em one} honest  miner expires.
	Since at least $\f{\Miners{}}+1$ miners are honest, they will only penalize $\Miner_i$ if they did not receive its nonce before $\delta$.
	This implies that either miner $\Miner_i$ did not find its nonce before $\delta$ (which is a direct contradiction of assumption (i)), or that its nonce did not reach the $\f{\Miners{}}+1$ honest miners before their timer expires (which is a direct contradiction of assumption (ii)).
\end{proof}

\section{Evaluation}
\label{s:eval}
Our evaluation aims to study the scalability and failure handling capability of \PoC{}. 

{\bf Implementation.}
We implement \PoC{}\footnote{https://anonymous.4open.science/r/poc-07E0/README.md} on top of the Apache ResilientDB, written in C++~\cite{apache-resdb};
\ResDB{} provides access to a scalable blockchain framework with APIs to implement and test new protocols.
It provides access to an optimized implementation of the \pbft{} consensus protocol.
Our \PoC{} implementation has an {\sf LOC} count of {\sf 1,300}, while \ResDB{} has {\sf 28,000 LOC}.

{\bf Setup.}
We run experiments on AWS c5.9xlarge (36 vCPUs and 72 GiB memory), up to $128$ miners and clients.
{\em Unless explicitly stated}, we use the following setup: 
deploy $128$ replicas in \ResDB{} to run \pbft{} consensus on $\Service$-blocks of size $100$.
In each experiment, first $20\%$ of the time we set as warmup and results are collected over the remaining time period.
We average results over {\em five} runs.
We use {\em ED25519}-based \DS{} for signing messages; \pbft{} makes use of {\em CMAC} for replica-to-replica communication.
Clients issue requests in a closed-loop; a client sends its next transaction only after it
receives response for its previous transaction.

{\bf Benchmark.}
We run two types of experiments: 
(1) impact of appending \PoC{} to \pbft{}.
(2) impact of appending \PoC{} to real-world blockchains.
For the first experiment, clients create  {\em YCSB}~\cite{ycsb,blockbench} transactions from the Blockbench~\cite{blockbench} framework. 
These single-operation transactions are key-value store operations that access a database of $\SI{600}{k}$ records.
For the second experiment, we extend the \Diablo{} benchmarking framework~\cite{diablo}, which 
issues client transactions to four state-of-the-art blockchain platforms.

\begin{figure}[t]
    \centering
    \setlength{\tabcolsep}{1pt}
    \begin{tabular}{cccc}
        \PbftBatchtput & \PbftBatchlat \\
        \pbfttput      & \pbftlat
    \end{tabular}
    \caption{Evaluation of \ResDB{} architecture.}
    \label{fig:resdb-eval}
\end{figure}

\subsection{Scalability of \ResDB{}}
\label{eval:pbft-batch-scale}
First, we illustrate the scalability of \ResDB{} fabric. 
This experiment serves as a baseline for the future experiments,  illustrating the impact of \PoC{} on the system throughput.
In Figures~\ref{fig:resdb-eval}(a) and (b), we measure the peak throughput (transactions per second) and latency for
\pbft{} consensus protocol on varying the block size ($\Service$-block) from $1$ to \SI{1}{k} on a system of $128$ replicas.

We observe that although \pbft{} hits its peak throughput at $\Service$-block size of $150$, 
the optimal $\Service$-block size is $100$ as the throughput is only $11.5\%$ less than the peak, while the latency is $3\times$ lower.

Beyond an $\Service$-block size of $150$, there is no increase in throughput because the queues that store messages at replicas are full and
can no longer process newer requests, which increases the wait time (latency) for clients.

Next, in Figures~\ref{fig:resdb-eval}(c)-(d), we increase the number of replicas from $16$ to $128$;
$\Service$-block size is $100$.
Unsurprisingly, on increasing the number of replicas,
there is a drop in the peak throughput (consequential increase in latency)
because there is a corresponding increase in the number of messages communicated per consensus; 
on moving from $16$ to $128$ replicas, the throughput drops by $86.8\%$.


\begin{figure}[t]
    \centering
    \scriptsize
    \begin{tabular}{|c|c|c|c|c|c|c|}
        \hline
        \makebox[0.05\textwidth]{\multirow{2}{*}{{\bf Protocol}}} & \multicolumn{3}{c|}{{\bf Mining Time (s)}} & \multicolumn{3}{c|}{{\bf Latency (s)}}                                                           \\ \cline{2-7}
                                        & \makebox[0.03\textwidth]{\bf $D=8$}                                & \makebox[0.03\textwidth][c]{\bf $D=9$}                            & \makebox[0.03\textwidth]{\bf $D=10$} & \makebox[0.03\textwidth]{\bf $D=8$} & \makebox[0.03\textwidth]{\bf $D=9$} & \makebox[0.03\textwidth]{\bf $D=10$} \\ \hline

        \makecell{\PoW{} (Sequential)}  & $696$                                      & $>2h$                                 & $>5h$          & $1233$       & $>3h$      & $>10h$         \\ \hline

        \PoW{} (Random)                          & $328$                                      & $4108$                                 & $>5h$          & $652$       & $8200$      & $>10h$         \\ \hline

        \makecell{\PoW{} + $30\%$ \PoC{}}          & $7$                                        & $165$                                  & $3103$       & $17$        & $303$       & $6950$       \\ \hline

        \rowcolor{green}
        \PoC{}                          & $3$                                        & $48$                                   & $738$        & $6$         & $84$        & $1260$       \\ \hline
    \end{tabular}
    \caption{Time to find a nonce by a system of $128$ miners while ensuring that these mining protocols meet \ResDB{} \pbft's throughput of $1$k blocks per second.}
    \label{tb:batch-size-poc}
\end{figure}

\begin{figure}[t]
    \centering
    \scriptsize
    \begin{tabular}{|c|c|c|}
        \hline
        \makebox[0.03\textwidth]{\bf $D=8$}                                & \makebox[0.03\textwidth][c]{\bf $D=9$}                            & \makebox[0.03\textwidth]{\bf $D=10$} \\ \hline
        5 possible solutions on average & 2 possible solutions on average & 1 solution  \\ \hline
    \end{tabular}

    \caption{Possible solutions on different difficulties.}
    \label{tb:poss-solution}
\end{figure}

\subsection{Scalability of \PoC{}}
\label{eval:poc-batch-scale}
In this set of experiments, we measure the time it takes for various mining schemes to append a block to the ledger.
In Figure~\ref{tb:batch-size-poc}, we compare \PoC{} against three baselines:
(1) Bitcoin's \PoW{} consensus, where each miner selects value uniformly at random and checks if it is a valid nonce. 
(2) \PoW{} consensus where each miner sequentially iterates over each value in the search space (starting from $0$) 
until it finds a valid nonce,
(3) \PoW{} with $30\%$ miners running \PoC{}, which allows us to approximate the impact of the largest centralized mining pool in Bitcoin~\cite{mining-pool-stats}.

In this experiment, we want to meet the following three goals:
(1) The mining scheme reaches the same throughput (blocks added to the ledger per second)
as \ResDB's throughput at $128$ replicas, which is approximately $100$k txns/s (or $1$k blocks per second).
(2) Aggregate precise number of $\Service$-blocks into a mined block so that we can observe least latency
(difference between time an $\Service$-block is received to the time it is added to the ledger).
(3) \PoC{} can address multiple solutions.
For example, for \PoC{}, the time to mine a block at $D=8$ is $3$s, so we add $3$k $\Service$-blocks in each mined block.
The latency for each experiment is approximately twice the mining time because
each transaction waits for a time equivalent to the mining time during the block generation phase and the mining process.
We observe that sequential \PoW{} mining requires at least $2\times$ more mining time and latency than randomized mining.
In comparison, \PoW{} with $30\%$ mining pool does reduce the mining time and latency.
However, at $D=10$, \PoC{} yields up to $4.2\times$ and $29\times$
less mining time than $30\%$ \PoC{} and \PoW{}, respectively.
In \Cref{tb:poss-solution},
we also demonstrate that multiple solutions may be located in different slices when $D$ is smaller. 
Like at $D=8$, each block contains around 5 solutions. 
However, \PoC{} ensures that each miner obtains the same solution to guarantee the safety of the block.

\begin{figure*}
    \centering
    \setlength{\tabcolsep}{2pt}
    \scalebox{0.82}{\ref{mainlegend}}\\[3pt]
    \begin{tabular}{ccccc}
        \DiabloNasTput & \DiabloVisaTput & \DiabloFifaTput & \DiabloDotaTput & \DiabloYoutubeTput \\
        \DiabloNasLat  & \DiabloVisaLat  & \DiabloFifaLat  & \DiabloDotaLat  & \DiabloYoutubeLat
    \end{tabular}
    \caption{Impact of appending \PoC{} (running at difficulty $D=8$) to different Diablo blockchains.}
    \label{fig:diablo_chains}
\end{figure*}

\begin{figure}[t]
    \centering
    \setlength{\tabcolsep}{1pt}
    \begin{tabular}{cc}
        \EvalMinerTput{a}   & \EvalMinerLatency{b}   \\
        \EvalLargeShiftTput & \EvalLargeShiftLatency
    \end{tabular}
    \caption{Peak throughput and average latency attained by \PoC{} under different conditions at $D=8$.}
    \label{fig:hybridchain-eval}
\end{figure}

Next, in Figures~\ref{fig:hybridchain-eval}(a) and (b), we increase the number of miners from $64$ to $128$ while fixing the difficulty, $D=8$.
For each setting, there is a specific size of mined block at which it sustains the throughput of \pbft{} and achieves least latency.
We observed $3$k, $5$k, and $7$k to be such block sizes.
We know that the peak throughput of \pbft{} is $100$k txns/s and latency at $D=8$ is around $6$s.
Thus, a \PoC{} protocol with $64$ miners hits the peak performance at $7$k while
a \PoC{} protocol with $96$ miners hits peak performance at $5$k.
Thus, we can conclude that \PoC{} is non-invasive and has minimal impact on the system throughput.

\subsection{Resilience to Failures}
\label{eval:fail}
Next, we illustrate the effect of failures on \PoC{}
by studying two types of failures: 1 malicious miner and No nonce (\S~\ref{ss:slice-shifting})
in Figures~\ref{fig:hybridchain-eval}(c) and (d).
%
In the malicious miner experiment, we simulate a Byzantine miner, which does not broadcast the solution of the $10$-th block,
which causes remaining miners to timeout and perform one round of the slice-shifting protocol.
In the no nonce experiment, we ensure that no nonce satisfies the $10$-th block, which leads to $\f{\Miners}+1$ rounds of slice shifting.
These experiments cause the latency to shoot up for the $10$-th block (up to $5\times$ for the no nonce case).
To quickly bring latency and throughput to the steady state, we merge the blocks in the no nonce case.


%

\subsection{Appending \PoC{} to Blockchains in the Wild.}
\label{eval:diablo}
Finally, we illustrate that \PoC{} can be integrated to state-of-the-art blockchain systems
to guard them against long-range attack with minimal impact.
We append \PoC{} to {\em four} blockchains part of the \Diablo~\cite{diablo} framework: \Diem{}~\cite{diembft}, 
\Algorand{}~\cite{algorand}, \Quorum{}~\cite{quorum} and \Ethereum{}~\cite{ether}.
Our primary goal is to showcase that appending \PoC{} to these blockchains
causes minimal impact on their performance.

\Diablo{} provides access to a
framework, written in Golang, for evaluating popular blockchain systems.
\Diablo{} is composed of three types of nodes: \textit{primary}, \textit{secondary}, and \textit{chain} nodes.
The primary node is responsible for generating transactions and delivering
them to the chain nodes via secondary nodes (mimicking a mempool functionality).
The secondary nodes assist the primary node in reducing resource overhead, e.g., CPU and network bandwidth,
by disseminating transactions to the chain nodes and collecting results. The chain nodes run various blockchains.

To incorporate \ResDB{} into Diablo framework,
we developed a Go SDK that provides an interface to send transactions from \Diablo{} to \ResDB{}.
To append \PoC{} to different chains, we implement a \texttt{Go-Server} agent on each chain node
that periodically fetches blocks from the local chain on the chain node through their SDKs implemented using Golang.
Specifically, the Go-Server provides a unified entry for \PoC{} to obtain the block data.
Note that in \Diablo{} clients submit requests in an open loop.

    {\bf Setup.}\label{eval:diablo_setup}
Like Diablo authors, we use AWS c5.9xlarge (36 vCPUs, 72 GiB memory) machines and deploy one primary and 10 secondaries.
We use the workloads provided by \Diablo{}.
The maximum observed throughput (in transactions per second or TPS) of these blockchains
as per the original paper is:
NasDAQ (168 TPS),
Visa (1000 TPS), Fifa (3483 TPS), Dota 2 (13303 TPS), and Youtube (37976 TPS).

    {\bf Results.}\label{eval:poc_chains}
Figure~\ref{fig:diablo_chains} summarizes our findings.
\Diem{} can sustain 1000 TPS and \Quorum{} only 204 TPS; their performance drops significantly when the workload increases.
\Ethereum{} attains a consistently low throughput in part due to its default block generation
period (15 seconds). \Algorand{} exhibits a more stable performance of
500-600 TPS across workloads. All blockchains suffer from higher latency (as expected) when
the load increases, in a sense, artificially over-saturating the system.
We observe that \ResDB{} can achieve a high throughput of 37,967 TPS, which nearly matching the injected load.
When \PoC{} is added to any system, we observe that the throughput drops
at most by 10\%-20\% due to an increased communication and added network latency.
We argue that this is a negligible cost as \PoC{} helps these blockchains prevent
long-range attacks, which continues to be a major vulnerability in their design.


\section{Related work}
\label{s:related}
\BFT{} has been studied extensively in the literature~\cite{pbftj,zyzzyva,iss,leaderless-consensus,malrec,blockchain-info-share,gem2tree,next700bft,xft,bft-forensics,lineagechain,loghin2022blockchain, chacko2023optimize, xiang2023practical, Wu25SigDetect, azouvi2022pikachu, deirmentzoglou2019survey, jain2023extending}. 
%
A sequence of efforts~\cite{pbftj,gueta2019sbft, hotstuff, kogias2016enhancing,bft-to-cft,occlum,shadoweth,eactors,bft-performance-framework,streamchain,borealis, li2018decentralized, chen2025blockchain, messadi2022splitbft, albouy2023asynchronous, zhang2019actor}
have been made to reduce the communication cost of the \BFT protocols:
(1) linearizing \BFT{} consensus~\cite{hotstuff,sbft}, 
(2) optimizing for geo-replication~\cite{steward}, and
(3) sharding~\cite{ahl,sharper,basil,el2019blockchaindb, amiri2019caper}.
Nevertheless, all of these protocols face long-range attacks~\cite{deirmentzoglou2019survey, azouvi2022pikachu}.

Alternatively, prior works have focussed on designing \PoS{} protocols that
permit the node with the highest stake to propose the next block~\cite{pos-anonymity,algorand,cardano,king2012ppcoin,ouroboros}.
However, even these protocols suffer from long-range attacks if adversary has access to the private keys.

Existing work to protect against long-range attacks includes: 
(1) checkpointing the state through a trusted committee~\cite{winkle,bitcoin-bitter-better,pos-weak-subjectivity,snow-white},
(2) key-evolving cryptographic techniques~\cite{pixel,ouroboros,algorand},
(3) verifiable delay functions~\cite{solana}, and
(4) appending the state to Bitcoin~\cite{bitcoin-pos,azouvi2022pikachu}. 

(1) {\em Trusted Committee} solutions aim to periodically checkpoint the state of a \PoS{} blockchain on another canonical chain, which is maintained by a 
committee of trusted members~\cite{winkle,bitcoin-bitter-better,pos-weak-subjectivity,snow-white, wang2020sperax}.
These systems assume that the trusted members cannot be compromised, and thus, 
new nodes that wish to join the system can distinguish between the \PoS{} blockchain and the canonical chain.
Moreover, small size committees mimic a centralized system, while large committees increase latency for checkpoints.

(2) {\em Key-evolving cryptographic techniques} force participants to periodically discard old keys and generate new keys~\cite{pixel,ouroboros,algorand, boyle2021breaking}.
These works assume that honest nodes will discard their old keys after they generate a new pair;
the onus is on the honest nodes.

(3) {\em Verifiable delay functions} provide proof that helps differentiate between a ledger created long ago versus a recently created adversarial ledger~\cite{solana, yu2024smart, choi2025smart}.
However, nothing prevents an adversary from initiating the creation of the adversarial ledger at the time of genesis. Once it has access to the private keys of other nodes, it can build blocks on top of this ledger, which makes it impossible for a new node to distinguish between the two ledgers.

(4) {\em Appending the state to Bitcoin} is a popular solution against long-range attacks for many recent papers~\cite{tas2022babylon,bitcoin-pos,azouvi2022pikachu}.
Bitcoin employs \PoW{} consensus, which requires miners to compete and thus wastes computational resources.
Prior solutions to improve \PoW{} or Bitcoin~\cite{byzcoin,pow-improve-1,bitcoinng} do not eliminate this competition.
%

With \PoC, we show how to make it computationally expensive for an adversary to rewrite the ledger while 
forcing miners to collaborate and conserve their computational resources.
Further, we show experimentally that collaboration helps \PoC{} to waste fewer resources for a given difficulty.

%

\section{Conclusions}
In this paper, we presented our novel \PoC{} protocol, which, when appended to existing \PoS/\BFT{} protocols, guards them against long-range attacks. 
Like \PoW{}, \PoC{} makes it computationally expensive for an adversary to rewrite the ledger. 
However, unlike \PoW{}, \PoC{} introduces collaborative mining that requires miners to work with each other instead of competing.

\section*{Acknowledgment}
  This work is partially funded by NSF Award Number 2245373.

\bibliographystyle{IEEEtran}
\bibliography{refined}

\newpage

\end{document}